\documentclass{lmcs} 
\pdfoutput=1
\usepackage[utf8]{inputenc}

\usepackage{lastpage}
\lmcsdoi{21}{4}{16}
\lmcsheading{}{\pageref{LastPage}}{}{}%
{Feb.~27,~2024}{Oct.~30,~2025}{}


\usepackage{amssymb}
\usepackage{proof}
\usepackage{booktabs}
\usepackage{graphicx}

\theoremstyle{plain} 

\newcommand{\ov}[1]{\overline{ #1 }}
\newcommand{\tld}[1]{\widetilde{ #1 }}
\newcommand{\wht}[1]{\widehat{ #1 }}

\newcommand{\MLM}{\mathbf{MLM}}
\newcommand{\MLMacc}{\mathbf{MLM_{acc}}}
\newcommand{\MLMext}{\mathbf{MLM_{ext}}}
\newcommand{\CZF}{\mathbf{CZF}}
\newcommand{\CZFM}{\mathbf{CZFM}}
\newcommand{\MLQ}{\mathbf{MLQ}}
\newcommand{\MLTT}{\mathbf{MLTT}}
\newcommand{\KPM}{\mathbf{KPM}}

\newcommand{\Id}{\mathrm{Id}}
\newcommand{\type}{\mathrm{type}}

\newcommand{\ctxt}{\mathrm{ctxt}}
\newcommand{\app}{\mathsf{app}}
\newcommand{\suc}{\mathsf{suc}}
\newcommand{\refl}{\mathsf{refl}}
\newcommand{\res}[2]{\mathsf{res}^{#1}_{#2}}
\newcommand{\Var}[1]{\mathrm{Var}(#1 )}
\newcommand{\ri}{\mathsf{i}}
\newcommand{\rj}{\mathsf{j}}
\newcommand{\rp}{\mathsf{p}}
\newcommand{\rE}{\mathsf{E}}
\newcommand{\rN}{\mathrm{N}}

\newcommand{\rW}{\mathrm{W}}

\newcommand{\rV}{\mathbb{M}}
\newcommand{\rTV}{\mathrm{T}_{\mathbb{M}}}
\newcommand{\rU}{\mathrm{U}}
\newcommand{\rT}{\mathrm{T}}
\newcommand{\Fam}[1]{\mathrm{Fam}( #1 )}
\newcommand{\rO}{\mathrm{O}}
\newcommand{\ru}{\mathrm{u}}

\newcommand{\TI}{\mathrm{TI}}
\newcommand{\idx}{\mathsf{index}}
\newcommand{\pred}{\mathsf{pred}}
\newcommand{\Acc}{\mathrm{Acc}}

\newcommand{\prog}{\mathsf{prog}}
\newcommand{\wcons}{\mathsf{sup}}

\newcommand{\fm}{\mathsf{fm}}
\newcommand{\intro}{\mathsf{in}}
\newcommand{\el}{\mathsf{el}}
\newcommand{\eq}{\mathsf{eq}}
\newcommand{\tc}[1]{#1_{\mathsf{tc}}}
\newcommand{\TC}[1]{\mathsf{TC}(#1)}

\newcommand{\erefl}{\mathsf{r}}

\newcommand{\bV}{\mathbf{V}}

\newcommand{\bbV}{\mathbb{V}}

\newcommand{\cV}{\mathcal{V}}

\newcommand{\iteM}[4]{\mathcal{U}^{#1 , #2}_{(#3 , #4)}}
\newcommand{\iteT}[4]{\mathrm{T}^{#1 , #2}_{(#3 , #4)}}
\newcommand{\iteV}[4]{\mathcal{V}^{#1 , #2}_{(#3 , #4)}}
\newcommand{\itebV}[4]{\mathbf{V}^{#1 , #2}_{(#3 , #4)}}
\newcommand{\iteh}[4]{\mathbf{h}^{#1 , #2}_{(#3 , #4)}}
\newcommand{\iteMh}[4]{\wht{\mathcal{U}}^{#1 , #2}_{(#3 , #4)}}
\newcommand{\iteTh}[4]{\wht{\mathrm{T}}^{#1 , #2}_{(#3 , #4)}}
\newcommand{\iteVh}[4]{\wht{\mathcal{V}}^{#1 , #2}_{(#3 , #4)}}
\newcommand{\iteHh}[4]{\Phi^{#1 , #2}_{(#3 , #4)}}

\def\eg{{\em e.g.}}
\def\cf{{\em cf.}}

\begin{document}

\title[Inaccessible Sets in Martin-L\"{o}f Type Theory with One Mahlo Universe]{Interpretation of Inaccessible Sets in\texorpdfstring{\\}{} Martin-L\"{o}f Type Theory with One Mahlo Universe}

\author[Y.~Takahashi]{Yuta Takahashi\lmcsorcid{0000-0002-5214-7077}}

\address{Aomori University, 2-3-1 Kobata, Aomori City, Aomori 030-0943, Japan}	
\email{y.takahashi@aomori-u.ac.jp}  






\begin{abstract}
  \noindent Rathjen proved that Aczel's constructive set theory $\CZF$ extended with inaccessible sets of all transfinite orders can be interpreted in Martin-L\"{o}f type theory $\MLTT$ extended with Setzer's Mahlo universe and another universe above it. In this paper we show that this interpretation can be carried out bottom-up without the universe above the Mahlo universe, provided we add an accessibility predicate instead. If we work in Martin-L\"{o}f type theory with extensional identity types the accessibility predicate can be defined in terms of $\rW$-types. The main part of our interpretation has been formalised in the proof assistant Agda.
\end{abstract}

\maketitle

\section{Introduction}\label{sec:intro}

\subsection*{Background and Aim}
Martin-L\"{o}f type theory $\MLTT$ was extended by Setzer \cite{setzer2000,setzer2008} with the so-called Mahlo universe types to provide a variant of intensional $\MLTT$ equipped with an analogue of Mahlo cardinals. Specifically, a Mahlo universe $\rV$ has a reflection property similar to that of Mahlo cardinals \cite{mahlo1911} (see also, \eg, \cite{drake1974}): for any function $f$ of type $\Sigma_{(x : \rV)}(\rTV x \to \rV) \to \Sigma_{(x : \rV)}(\rTV x \to \rV)$, where $\rTV$ is the decoding function for $\rV$, the universe $\rV$ contains the subuniverse which is closed under $f$ and denoted by $\rU_f$. The resulting system, $\MLM$, is thus an instance of a constructive system extended with an analogue of Mahlo cardinals. Setzer's purpose of introducing $\MLM$ is to obtain an extension of $\MLTT$ whose proof-theoretic ordinal is slightly greater than that of $\KPM$, \textit{i.e.}, Kripke-Platek set theory with one recursively Mahlo ordinal, where the proof-theoretic ordinals of a system enable measurement of the proof-theoretic strength. The proof-theoretic ordinal of $\MLM$ was determined by the results of \cite{setzer2000,setzer2008}.

Another instance of constructive systems extended with an analogue of Mahlo cardinals was developed in the context of Aczel's constructive set theory: $\CZF$ \cite{aczel1978,aczel1982,aczel1986}. Rathjen \cite{rathjen2003} formulated the system $\CZFM$ by extending $\CZF$ with the axiom asserting the existence of one Mahlo set. He then verified that $\CZFM$ is interpretable in an extension of $\MLM$ which has one usual Tarski universe above the Mahlo universe $\rV$. Though he calls this extension $\MLM$, Setzer's $\MLM$ does not have any universe above $\rV$ and so we call this extension $\MLM^+$ to distinguish it from $\MLM$. The main purpose of introducing $\CZFM$ is a proof-theoretic one, as in the case of Mahlo universes. The interpretation of $\CZFM$ in $\MLM^+$ implies that $\CZFM$ has at most the proof-theoretic strength of $\MLM^+$.

On the other hand, Rathjen, Griffor and Palmgren \cite{RathjenGrifforPalmgren1998} introduced the system $\CZF_{\pi}$, which is an extension of $\CZF$ with the inaccessible sets of all transfinite orders \cite{mahlo1911}. They also introduced an extension of $\MLTT$ with second-order universe operators called $\MLQ$, where the base theory of $\MLQ$ is extensional $\MLTT$ in \cite{martinloef1984}. Then, Rathjen, Griffor and Palmgren showed that $\CZF_{\pi}$ has at most the proof-theoretic strength of $\MLQ$ by verifying that $\CZF_{\pi}$ is interpretable in $\MLQ$. As to the strength of large sets concerned, $\CZF_{\pi}$ is weaker than $\CZFM$: this claim is verified by Rathjen in \cite{rathjen2003}, where he proved that $\CZF_{\pi}$ is interpretable in $\CZFM$.

In summary, Rathjen \cite{rathjen2003} have shown that $\CZF_{\pi}$ can be interpreted in $\MLM^+$ by proving that $\CZF_{\pi}$ and $\CZFM$ are interpretable in $\CZFM$ and $\MLM^+$, respectively. However, it was unknown whether a universe above the Mahlo universe $\rV$ is indispensable for the interpretation of $\CZF_{\pi}$ in $\MLTT$ with the Mahlo universe $\rV$. Rathjen used the universe above $\rV$ to define the type of Aczel's iterative sets, which is a model of $\CZFM$ in $\MLM^+$. So the universe above $\rV$ is crucial to his argument, and adding one universe type makes a system stronger as shown by Setzer \cite{setzer1998,setzer1993}. Are the inaccessible sets of $\CZF_{\pi}$ interpretable by $\rV$ without a universe above it?

In this paper, we show that $\CZF_{\pi}$ is interpretable in $\MLM$ with the accessibility predicate $\Acc$: one can dispense with a universe above the Mahlo universe, but the trade-off is the use of $\Acc$. The accessibility predicate $\Acc$ provides the transfinite induction principle on Aczel's iterative sets. This predicate can be defined in Dybjer-Setzer's theory of induction-recursion \cite{DybjerSetzer2006}, which is a finite axiomatisation of Dybjer's general schema of simultaneous inductive-recursive definitions \cite{dybjer2000}. Specifically, one can define $\Acc$ by indexed inductive definition. To the best of our knowledge, $\MLTT$ without the schema of indexed inductive definition is not sufficient to define indexed inductive types such as $\Acc$, though some progress was made by, \eg, Hugunin \cite{hugunin2021}. We thus add $\Acc$ as a base type to $\MLM$, while we show in Appendix~\ref{sec:accproof} that \textit{extensional} $\MLTT$ in \cite{martinloef1984} can define $\Acc$. So, in the extensional setting, $\Acc$ is not necessary in our interpretation of $\CZF_{\pi}$ as in that of \cite{RathjenGrifforPalmgren1998}.

Contrary to Rathjen's elegant but ``top-down'' argument, the Mahlo universe $\rV$ together with $\Acc$ enables to provide a ``bottom-up'' argument to interpret $\CZF_{\pi}$. Below we first define a higher-order (in fact, second-order) universe operator using the reflection property of $\rV$ in $\MLM$. Next, the transfinite hierarchy of Aczel's iterative sets is constructed by iterating this universe operator along the accessibility given by $\Acc$. Finally, we verify that these sets are indeed type-theoretic counterparts of inaccessible sets. Besides this aim, we provide an explicit construction of super universes as subuniverses of $\rV$ for the expository purpose. In particular, it will be shown that the version of super-universe construction which is given by the higher-order universe operator above can be iterated.

Moreover, we formalised the main part of our proof in the proof assistant Agda \cite{AgdaTeam2024}. The Mahlo universe $\rV$ is simulated by the \textit{external} Mahlo universe in Agda: the external Mahlo universe is a proof-theoretically weak variant of $\rV$, and can be defined with induction-recursion as Dybjer and Setzer have shown \cite{DybjerSetzer2003}. So Agda is equipped with the external Mahlo universe because Agda supports even indexed induction-recursion, which is a generalisation of induction-recursion. This means that Agda also has the accessibility predicate $\Acc$. In Agda we defined the higher-order universe operator above by using the reflection property of the external Mahlo universe, and formulated our type-theoretic interpretation of $\CZF_{\pi}$. At the end of Section~\ref{sec:inacc}, we will comment on the part of our proof which is not formalised yet. 

Note that neither our argument in this paper nor our formalisation postulates extensional principles such as the rules of extensional identity types, except that Appendix~\ref{sec:accproof} derives all inference rules for $\Acc$ by using extensional $\MLTT$ in \cite{martinloef1984}. The Agda code for the formalisation is available online \cite{takahashi2025}. The clickable symbol \texttt{Link}, which is attached to some statements in the main part of our proof, will open the Agda code in \cite{takahashi2025} corresponding to those statements. All files in this repository can be typechecked with Agda version 2.7.0.1.

In Table~\ref{tab:systems} we list all of the systems discussed in this paper.
\begin{table}[t]
  \centering
  \scalebox{0.95}{
  \begin{tabular}{@{}ll@{}}
    \toprule
    \multicolumn{1}{@{}c@{}}{Acronym} & \multicolumn{1}{@{}c@{}}{Description} \\
    \midrule
    $\MLM$ & intensional Martin-L\"{o}f type theory with Setzer's Mahlo universe \cite{setzer2000,setzer2008}\\
    $\MLM^+$ & $\MLM$ with a universe above its Mahlo universe \cite{rathjen2003} \\
    $\MLQ$ & extensional Martin-L\"{o}f type theory with second-order universe operators \cite{RathjenGrifforPalmgren1998} \\
    $\MLMacc$ & $\MLM$ with an accessibility predicate \\
    $\MLMext$ & $\MLM$ with extensional identity types \\
    $\CZF$ & Aczel's constructive Zermelo-Fraenkel set theory \cite{aczel1978,aczel1982,aczel1986}\\
    $\CZF_{\pi}$ & $\CZF$ with inaccessible sets of all transfinite orders \cite{RathjenGrifforPalmgren1998} \\
    $\CZFM$ & $\CZF$ with one Mahlo set \cite{rathjen2003} \\
    \bottomrule
  \end{tabular}
  }
  \caption{The Systems Discussed in the Present Paper}
  \label{tab:systems}
\end{table}

\subsection*{Related Work}
The interpretation of $\CZF$ in $\MLTT$ was first discussed in \cite{aczel1978}. Here Aczel set out $\CZF$ and interpreted it in the fragment $\mathbf{ML_1 V}$ of $\MLTT$, which has dependent function types (\textit{i.e.}, $\Pi$-types), dependent pair types (\textit{i.e.}, $\Sigma$-types), disjoint union types, the empty type, the unit type, the natural number type, the universe type and an instance of $\rW$-types. Iterative sets in $\CZF$ were interpreted by this instance of $\rW$-types. Later, the set-theoretic axiom of dependent choice was interpreted in the extension of $\mathbf{ML_1 V}$ by extensional identity types \cite{aczel1982}. Moreover, the regular extension axiom ($\mathbf{REA}$), which states that any set is a subset of a regular set, was interpreted by further adding $\rW$-types \cite{aczel1986}.

Griffor and Rathjen \cite{GrifforRathjen1994} showed that $\CZF$ and $\mathbf{ML_1 V}$ have the same proof-theoretic strength and there is a fragment of $\MLTT$ with the same proof-theoretic strength as $\CZF + \mathbf{REA}$. Setzer \cite{setzer1998,setzer1993} determined the proof-theoretic strength of the fragment of $\MLTT$ obtained from $\mathbf{ML_1 V}$ by adding intensional identity types and $\rW$-types: he computed its proof-theoretic ordinal without appealing to Aczel's interpretation of $\CZF$ in $\MLTT$.

As aforementioned, Rathjen, Griffor and Palmgren \cite{RathjenGrifforPalmgren1998} introduced the extension $\CZF_{\pi}$ of $\CZF$ by the axiom asserting the existence of inaccessible sets of all transfinite orders. The type system in which $\CZF_{\pi}$ was interpreted is $\MLQ$, which is obtained from extensional $\MLTT$ in \cite{martinloef1984} by adding the inductive type $\mathrm{Q}$ and universe type $\mathrm{M}$: $\mathrm{Q}$ is an inductive type of codes for operators which provides universes closed under universe operators constructed previously, and $\mathrm{M}$ is a universe closed under the operators whose codes belong to $\mathrm{Q}$. Our higher-order universe operator is inspired by the one defined with the types $\mathrm{Q}$ and $\mathrm{M}$, and our interpretation of $\CZF_{\pi}$ is the same as that of \cite{RathjenGrifforPalmgren1998}, except for the interpretation of inaccessible sets. While $\MLM$ is proof-theoretically stronger than $\MLQ$, it is an open question whether the intensional variant of $\MLQ$ can be interpreted in $\MLM$. There is a difficulty in simulating the type of \textit{codes} for operators of $\MLQ$ by $\MLM$ because such a type of codes for higher-order objects is absent from $\MLM$. Rathjen \cite{rathjen2000} introduced the fragment $\mathbf{MLF}$ of $\MLQ$ and showed that two weaker variants of $\CZF_{\pi}$ have the same proof-theoretic strength as $\mathbf{MLF}$.

The relationship between $\MLTT$ and $\CZF$ in the context of an analogue for Mahlo cardinals has already been examined by Rathjen \cite{rathjen2003}, as shown above. He introduced the notion of a Mahlo set into $\CZF$ and verified that the resulting system $\CZFM$ is interpretable in $\MLM^+$, which is an extension of $\MLM$ with one Tarski universe above the Mahlo universe. He also showed that $\CZF_{\pi}$ is interpretable in $\CZFM$, but this result is not explicitly stated and so we briefly recall it in Section~\ref{sec:czf}.

\subsection*{Contributions}
Our contributions are summarised as follows:
\begin{enumerate}
\item The principal contribution of this paper is the construction of inaccessible sets of transfinite orders in $\MLM$ with the accessibility predicate $\Acc$. To describe this interpretation, we use both the reflection property of $\rV$ and the accessibility predicate $\Acc$. We define a higher-order universe operator with the reflection property of $\rV$ and iterate this operator along the accessibility given by $\Acc$.

\item A secondary contribution is the Agda formalisation of the main proof for the result above. We utilise indexed induction-recursion in Agda for formalisation: the Mahlo universe of $\MLM$ is simulated by the external Mahlo universe defined with induction-recursion, and the predicate $\Acc$ is defined by using indexed inductive definition.
  
\end{enumerate}

\subsection*{Organisation of the Paper}
In Section~\ref{sec:unim}, we describe the system $\MLM$ and define a higher-order universe operator using the reflection property of the Mahlo universe $\rV$. Section~\ref{sec:large} first provides a brief introduction to $\CZF_{\pi}$ and then shows that $\CZF_{\pi}$ is interpretable in $\MLM$ with the accessibility predicate $\Acc$. Specifically, we iterate the application of our higher-order operator along the accessibility given by $\Acc$ to construct the hierarchy of inaccessible sets in $\MLM$.


\section{A Higher-Order Universe Operator by the Mahlo Universe in $\MLM$}\label{sec:unim}
In this section, we first present the system $\MLM$ of Martin-L\"{o}f type theory with one Mahlo universe (Section~\ref{sec:mlm}). Next, we formulate a higher-order universe operator using the reflection property of the Mahlo universe (Section~\ref{sec:high}). This operator will be used to construct the transfinite hierarchy of inaccessible sets.


\subsection{Martin-L\"{o}f Type Theory with One Mahlo Universe}\label{sec:mlm}
We follow Setzer's work \cite{setzer2000,setzer2008,setzer2008a} in defining the rules of universes in the system $\MLM$ of Martin-L\"{o}f type theory with one Mahlo universe, while we follow \cite{hofmann1997,setzer1998} in defining the inference rules other than those of universes. The system $\MLM$ is an intensional $\MLTT$ (\textit{i.e.}, an $\MLTT$ with the intensional identity types) formulated in the polymorphic way. See \cite{NPS1990,DybjerPalmgren2020} for the distinction between intensional $\MLTT$'s and extensional ones; see \cite{NPS1990} for the distinction between the polymorphic formulation and the monomorphic one.

In this paper, syntactic variables are used possibly with suffixes. Let an infinite set of variables be given. Variables are denoted by $x,y,z,v,w$, and numerals $0,1,\ldots$ are denoted by $n , m , k$. \textit{Pre-terms} $a , b , c , \ldots$ and \textit{pre-types} $A , B , C , \ldots$ are defined by simultaneous induction in Table \ref{tab:pretypes}. The intuitive meanings of pre-terms and pre-types will be explained by the inference rules of $\MLM$.
\begin{table}[t]
\begin{align*}
  a , b , c , d , f , g , h , l ::= &\: x \mid 0 \mid m_n \mid \lambda x . a \mid \app ( a , b ) \mid ( a , b ) \mid \rE_{\Sigma} ( a , b ) \mid  \\
  &\: \ri ( a ) \mid \rj ( a ) \mid \rE_+ ( a , b , c ) \mid \suc ( a ) \mid \rE_{\rN} ( a , b , c ) \mid \\
  &\: \rE_n ( a_0 , \ldots , a_n ) \mid \wcons ( a , b ) \mid \rE_{\rW} ( a , b ) \mid \refl ( a ) \mid \rE_{\Id} ( a , b ) \mid \\
  &\: \wht{\rN}_{\rV} \mid \wht{\rN_n}_{\rV} \mid \wht{\Pi}_{\rV} ( a , b ) \mid \wht{\Sigma}_{\rV} ( a , b ) \mid a \wht{+}_{\rV} b \mid \wht{\rW}_{\rV} ( a , b ) \mid \\
  &\:  \wht{\Id}_{\rV} ( a , b , c ) \mid \wht{\rU}_{f} \mid \wht{\rT}_{f} (a) \mid \res{0}{f} ( a , b ) \mid \res{1}{f} ( a , b , c ) \\
  &\: \wht{\rN}_{f} \mid \wht{\rN_n}_{f} \mid \wht{\Pi}_{f} ( a , b ) \mid \wht{\Sigma}_{f} ( a , b ) \mid a \wht{+}_{f} b \mid \wht{\rW}_{f} ( a , b ) \mid \wht{\Id}_{f} ( a , b , c ) \\
  A , B , C , D ::= &\: \rN \mid \rN_n \mid \rV \mid \rU_{f} \mid \Pi_{(x : A)}B \mid \Sigma_{(x : A)}B \mid A + B \mid \\
  &\: \rW_{(x : A)} B \mid \Id ( A , a , b ) \mid \rTV (a)
\end{align*}
\caption{Pre-Terms and Pre-Types}
\label{tab:pretypes}
\end{table}
The pre-terms of the form $\app (a , b)$ are abbreviated as $a \: b$. When a variable $x$ is not free in $B$, we write $\Pi_{(x : A)} B$ as $A \to B$, and $\Sigma_{(x : A)} B$ as $A \times B$. The occurrences of $x$ are bound in the expressions of the form $\Pi_{(x : A)} B , \Sigma_{(x : A)} B , \rW_{(x : A)} B$ or $\lambda x . a$. Consecutive applications $k_1 (k_2 (\cdots k_n (a) \cdots ))$ of $\ri , \rj$ with $k_1 , k_2 , \ldots , k_n \in \{ \ri , \rj \}$ are abbreviated as $k_1 k_2 \cdots k_n a$. Any two $\alpha$-equivalent expressions are identified. In addition, we adopt Barendregt's variable condition: no variable occurs both as a free one and as a bound one in an expression, and all bound variables in an expression are distinct. We denote by $e [e_1 / x_1 ,\ldots , e_n / x_n]$ the expression obtained by substituting $e_1 ,\ldots , e_n$ for $x_1 ,\ldots , x_n$ in $e$ consecutively. We always assume that no bound variable in $e$ occurs freely in any of $e_1 , \ldots , e_n$, by applying the $\alpha$-conversion if necessary.

\textit{Judgements} are expressions which have one of the following forms:
\begin{itemize}
\item $\vdash \Gamma \; \ctxt$, which means that $\Gamma$ is a valid context,

\item $\Gamma \vdash A \; \type$, which means that $A$ is a type in a context $\Gamma$,

\item $\Gamma \vdash a : A$, which means that $a$ is a term of type $A$ in a context $\Gamma$,

\item $\Gamma \vdash A = B$, which means that $A$ and $B$ are definitionally equal types in a context $\Gamma$,

\item $\Gamma \vdash a = b : A$, which means that $a$ and $b$ are definitionally equal terms of type $A$ in a context $\Gamma$.
\end{itemize}
According to the propositions-as-types conception, $\Gamma \vdash a : A$ means that $a$ is a proof of the proposition $A$ in the context $\Gamma$. One can consider two definitionally equal expressions to be two expressions whose results of computation are equal. In an extensional $\MLTT$ a judgement of the form $\Gamma \vdash A = B$ or $\Gamma \vdash a = b : A$ no longer means that the two expressions mentioned are definitionally equal: an extensional $\MLTT$ treats these judgement forms in a more extensional manner (for the rules of extensional identity types, see Appendix~\ref{sec:accproof}). We denote by $\theta$ an expression of one of the following forms:
\[
A \; \type , \qquad a : A , \qquad A = B , \qquad a = b : A .
\]

The system $\MLM$ is a formal system to derive judgements of the forms above by means of inference rules. Its inference rules are classified as follows.
\begin{itemize}
\item General rules: those of contexts, assumptions, substitutions and the definitional equality,

\item Rules for standard types: those of $\Pi$-types, $\Sigma$-types, $+$-types, $\rN$-type, $\rN_n$-types, $\rW$-types and $\Id$-types,

\item Rules for universe types: those of Mahlo universe types and non-Mahlo universe types.
\end{itemize}
Each collection of rules for standard types includes the \textit{formation rule}, the \textit{introduction rules}, the \textit{elimination rules} and the \textit{equality rules}. The formation rule explains how to form the types, and the introduction rules enable the construction of canonical elements of the types. On the other hand, the elimination rules provide executable elements of the types, and the equality rules specify how to execute (or compute) these elements in an equational way. Here we present the rules for Mahlo universe only, and the other rules of $\MLM$ are presented in Appendix~\ref{sec:mlmrules}. For detailed explanations of the above rules except those of Mahlo universes, see, \eg, \cite{martinloef1984,NPS1990,hofmann1997,DybjerPalmgren2020}. Below we denote the Mahlo universe in $\MLM$ by $\rV$, while it is denoted by $\mathrm{V}$ in \cite{setzer2000,setzer2008,setzer2008a}.

For the Mahlo Universe $(\rV , \rTV)$, we have the $\rV$-formation rule $\rV \fm$, the $\rV$-introduction rules $\rV \intro$ and the $\rTV$-formation rule $\rTV \fm$ with the inference rules which define the decoding function $\rTV$ recursively. Congruence rules are defined for the $\rTV$-formation rule and all $\rV$-introduction rules except $\rV \intro_{\rN}$ and $\rV \intro_{\rN_n}$, though we omit to write down them. Some examples of congruence rules can be found in the list of the rules for $\Pi$-types in Appendix~\ref{sec:mlmrules}.
    \[
    \infer[\rV \fm]{\Gamma \vdash \rV \; \type}{\vdash \Gamma \; \ctxt}
    \qquad
    \infer[\rTV \fm]{\Gamma \vdash \rTV (a) \; \type}{\Gamma \vdash a : \rV}
    \]
    \[
    \infer[\rV \intro_{\rN}]{\Gamma \vdash \wht{\rN}_{\rV} : \rV}{\vdash \Gamma \; \ctxt}
    \qquad
    \infer{\Gamma \vdash \rTV ( \wht{\rN}_{\rV} ) = \rN}{\vdash \Gamma \; \ctxt}
    \]
    \[
    \infer[\rV \intro_{\rN_n}]{\Gamma \vdash \wht{\rN_n}_{\rV} : \rV}{\vdash \Gamma \; \ctxt}
    \qquad
    \infer{\Gamma \vdash \rTV ( \wht{\rN_n}_{\rV} ) = \rN_n}{\vdash \Gamma \; \ctxt}
    \]
    \[
    \infer[\rV \intro_{\Pi}]{\Gamma \vdash \wht{\Pi}_{\rV} ( a , b ) : \rV}{
      \deduce{\Gamma \vdash b : \rTV ( a ) \to \rV}{\Gamma \vdash a : \rV}
    }
    \qquad
    \infer{\Gamma \vdash \rTV ( \wht{\Pi}_{\rV} ( a , b ) ) = \Pi_{(x : \rTV ( a ))} \rTV ( b \: x ) }{
      \deduce{\Gamma \vdash b : \rTV ( a ) \to \rV}{\Gamma \vdash a : \rV}
    }
    \]
    \[
    \infer[\rV \intro_{\Sigma}]{\Gamma \vdash \wht{\Sigma}_{\rV} ( a , b ) : \rV}{
      \deduce{\Gamma \vdash  b : \rTV ( a ) \to \rV}{\Gamma \vdash a : \rV}
    }
    \qquad
    \infer{\Gamma \vdash \rTV ( \wht{\Sigma}_{\rV} ( a , b ) ) = \Sigma_{(x : \rTV ( a ))} \rTV ( b \: x) }{
      \deduce{\Gamma \vdash b : \rTV ( a ) \to \rV}{\Gamma \vdash a : \rV}
    }
    \]
    \[
    \infer[\rV \intro_{+}]{\Gamma \vdash a \wht{+}_{\rV} b : \rV}{
      \deduce{\Gamma \vdash b : \rV}{\Gamma \vdash a : \rV}
    }
    \qquad
    \infer{\Gamma \vdash \rTV ( a \wht{+}_{\rV} b ) = \rTV ( a ) + \rTV ( b ) }{
      \deduce{\Gamma \vdash b : \rV}{\Gamma \vdash a : \rV}
    }
    \]
    \[
    \infer[\rV \intro_{\rW}]{\Gamma \vdash \wht{\rW}_{\rV} ( a , b ) : \rV}{
      \deduce{\Gamma \vdash b : \rTV ( a ) \to \rV}{\Gamma \vdash a : \rV}
    }
    \qquad
    \infer{\Gamma \vdash \rTV ( \wht{\rW}_{\rV} ( a , b ) ) = \rW_{(x : \rTV ( a ))} \rTV ( b \: x) }{
      \deduce{\Gamma \vdash b : \rTV ( a ) \to \rV}{\Gamma \vdash a : \rV}
    }
    \]
    \[
    \infer[\rV \intro_{\Id}]{\Gamma \vdash \wht{\Id}_{\rV} ( a , b , c ) : \rV}{
      \deduce{\Gamma \vdash c : \rTV ( a )}{
        \deduce{\Gamma \vdash b : \rTV ( a )}{\Gamma \vdash a : \rV}
      }
    }
    \qquad
    \infer{\Gamma \vdash \rTV ( \wht{\Id}_{\rV} ( a , b , c ) ) = \Id ( \rTV ( a ) , b , c  ) }{
      \deduce{\Gamma \vdash c : \rTV ( a )}{
        \deduce{\Gamma \vdash b : \rTV ( a )}{\Gamma \vdash a : \rV}
      }
    }
    \]
\pagebreak[5] 

The rules concerning the codes $\wht{\rU}_{f}$ for subuniverses of $\rV$ are also included as follows.
    \[
    \infer[\rV \intro_{\rU}]{\Gamma \vdash \wht{\rU}_{f} : \rV}{
      \Gamma \vdash f : \Sigma_{(x : \rV)} (\rTV ( x ) \to \rV) \to \Sigma_{(x : \rV)} (\rTV ( x ) \to \rV)
    }
    \]
    \[
    \infer[\rTV \mathsf{uni}]{\Gamma  \vdash \rTV ( \wht{\rU}_{f} ) = \rU_{f}}{
      \Gamma \vdash f : \Sigma_{(x : \rV)} (\rTV ( x ) \to \rV) \to \Sigma_{(x : \rV)} (\rTV ( x ) \to \rV)
    }
    \]
    \[
    \infer[\rV \intro_{\wht{\rT}}]{\Gamma  \vdash \wht{\rT}_{f} (a) : \rV}{
      \deduce{\Gamma  \vdash a : \rU_{f}}{
        \Gamma \vdash f : \Sigma_{(x : \rV)} (\rTV ( x ) \to \rV) \to \Sigma_{(x : \rV)} (\rTV ( x ) \to \rV)
      }
    }
    \]
    
As we have seen in Section~\ref{sec:intro}, the Mahlo universe $\rV$ reflects functions on families of small types in $\rV$. The system $\MLM$ expresses this feature as follows. First of all, the $\rV \intro_{\rU}$ rule states that there is a subuniverse $\rU_{f}$ of $\rV$ closed under $f : \Sigma_{(x : \rV)} (\rTV (x) \to \rV ) \to \Sigma_{(x : \rV)} (\rTV (x) \to \rV )$ with its code $\wht{\rU}_{f} : \rV$. The closure under $f$ is then expressed by $\res{0}{f}$ and $\res{1}{f}$, which provide the restriction of $f$ to $\rU_{f}$: we have $\res{0}{f} ( a , b ) : \rU_{f}$ and $\res{1}{f} ( a , b , c) : \rU_{f}$ for any $a : \rU_{f}$, $b : \rT_{f} ( a ) \to \rU_{f}$ and $c : \rT_{f} (\res{0}{f} ( a , b ))$. The following computation rules given by the corresponding inference rules of $\MLM$ show that $\res{0}{f}$ and $\res{1}{f}$ actually provide the restriction of $f$ to $\rU_{f}$ via the ``injection'' $\wht{\rT}_{f}$ from $\rU_{f}$ to $\rV$.
\begin{align*}
  \wht{\rT}_{f} ( \res{0}{f} ( a , b ) ) &= \rp_1 (f \: (\wht{\rT}_{f} ( a ) , \lambda x . \wht{\rT}_{f} (b \: x))) : \rV \\
  \wht{\rT}_{f} ( \res{1}{f} ( a , b , c ) ) &= \rp_2 (f \: (\wht{\rT}_{f} ( a ) , \lambda x . \wht{\rT}_{f} ( b \: x))) \: c : \rV
\end{align*}
That is, $\wht{\rT}_{f}$ maps $\res{0}{f} ( a , b )$ to the left projection of $f \: (\wht{\rT}_{f} ( a ) , \lambda x . \wht{\rT}_{f} (b \: x))$, and $\res{1}{f} ( a , b , c )$ to the right projection of $f \: (\wht{\rT}_{f} ( a ) , \lambda x . \wht{\rT}_{f} (b \: x))$ applied to $c$.

We work informally in $\MLM$ to improve readability. For any context $\Gamma$, put
\[
\Var{\Gamma} := \{ x \mid \text{$x : A$ belongs to $\Gamma$ for some type $A$} \}.
\]
The $\Sigma$-type $\Sigma_{(x : \rV)} (\rTV (x) \to \rV )$ is abbreviated as $\Fam{\rV}$. The type constructor $+$ for sum types is treated as left associative, \textit{i.e.},
\[
A_1 + A_2 + \cdots + A_n := (\cdots (A_1 + A_2 ) + \cdots + A_n ).
\]
Below we write $\Id (A , a , b )$ as $a =_{A} b$.
\begin{exa}[Universes above Families of Small Types]\label{exa:uni}
  The Mahlo universe $\rV$ provides universes above families of small types as follows. Let $a_0 : \rV$ and $b_0 : \rTV (a_0) \to \rV$ be given, and assume that neither $a_0$ nor $b_0$ has any free occurrence of a variable $x$. Defining $f_0 : \Fam{\rV} \to \Fam{\rV}$ as $f_0 := \lambda x . (a_0 , b_0)$, we have the universe $\rU_{f_0}$ with the decoding function $\wht{\rT}_{f_0}$. The universe $\rU_{f_0}$ can be considered as a universe above $(a_0 , b_0)$, as the following equations show:
  \begin{align*}
    \wht{\rT}_{f_0} ( \res{0}{f_0} ( \wht{\rN_0}_{f_0} , \lambda v . \rE_0 (v) ) ) &= \rp_1 \Bigl( f_0 \: (\wht{\rT}_{f_0} ( \wht{\rN_0}_{f_0} ) , \lambda v . \wht{\rT}_{f_0} (\rE_0 (v))) \Bigr) = \rp_1 ((a_0 , b_0)) = a_0 \\
    \wht{\rT}_{f_0} ( \res{1}{f_0} ( \wht{\rN_0}_{f_0} , \lambda v . \rE_0 (v) , c ) ) &= \rp_2 \Bigl( f_0 \: (\wht{\rT}_{f_0} ( \wht{\rN_0}_{f_0} ) , \lambda v . \wht{\rT}_{f_0} ( \rE_0 (v) )) \Bigr) \: c = \rp_2 ((a_0 , b_0)) \: c = b_0 \: c
  \end{align*}
  That is, $\res{0}{f_0} ( \wht{\rN_0}_{f_0} , \lambda v . \rE_0 (v) )$ is a code for $a_0$ in $\rU_{f_0}$, and $\res{1}{f_0} ( \wht{\rN_0}_{f_0} , \lambda v . \rE_0 (v) , c )$ is a code for $b_0 \: c$ in $\rU_{f_0}$, where $\wht{\rN_0}_{f_0}$ and $\lambda v . \rE_0 (v)$ are dummy arguments.
\end{exa}

\begin{exa}[Super Universe]\label{exa:suni}
  A super universe, which is a universe closed under a universe operator, can be defined by reflecting a function with a parameter. Put $f_1 : \Fam{\rV} \to \Fam{\rV}$ as $f_1 := \lambda z . y$, where the free variable (or the parameter) $y$ is of type $\Fam{\rV}$. We have $(\wht{\rU}_{f_1} , \wht{\rT}_{f_1}) : \Fam{\rV}$ by reflecting $f_1$. Then, we define $g := \lambda y . (\wht{\rU}_{f_1} , \wht{\rT}_{f_1})$, and obtain the universe $\rU_{g}$ by reflecting $g$. Note that $g : \Fam{\rV} \to \Fam{\rV}$ is a universe operator because it takes a family of small types in $\rV$ and returns a universe above this family. The universe $\rU_{g}$ is thus a super universe, because $\res{0}{g}$ and $\res{1}{g}$ show that $\rU_{g}$ is actually closed under $g$ as the equations below show: for any $a : \rU_{g}$ and $b : \rT_{g} (a) \to \rU_{g}$,
  \begin{align*}
    \wht{\rT}_{g} ( \res{0}{g} ( a , b ) ) &= \rp_1 \Bigl( g \: (\wht{\rT}_{g} ( a ) , \lambda v . \wht{\rT}_{g} (b \: v)) \Bigr) \\
    &= \rp_1 (\wht{\rU}_{\lambda z . (\wht{\rT}_{g} ( a ) , \lambda v . \wht{\rT}_{g} (b \: v))} , \wht{\rT}_{\lambda z . (\wht{\rT}_{g} ( a ) , \lambda v . \wht{\rT}_{g} (b \: v))}) \\
    &= \wht{\rU}_{\lambda z . (\wht{\rT}_{g} ( a ) , \lambda v . \wht{\rT}_{g} (b \: v))} \\
    \wht{\rT}_{g} ( \res{1}{g} ( a , b , c ) ) &= \rp_2 \Bigl( g \: (\wht{\rT}_{g} (a) , \lambda v . \wht{\rT}_{g} (b \: v )) \Bigr) \: c \\
    &= \rp_2 (\wht{\rU}_{\lambda z . (\wht{\rT}_{g} ( a ) , \lambda v . \wht{\rT}_{g} (b \: v))} , \wht{\rT}_{\lambda z . (\wht{\rT}_{g} ( a ) , \lambda v . \wht{\rT}_{g} (b \: v))}) \: c \\
    &= \wht{\rT}_{\lambda z . (\wht{\rT}_{g} ( a ) , \lambda v . \wht{\rT}_{g} (b \: v))} \: c
  \end{align*}
  That is, $\res{0}{g} ( a , b )$ is a code of the universe above the family $(a , b)$, and $\res{1}{g} ( a , b )$ is a code of its decoding function. Notice that $\rU_g$ also has a code $\res{0}{g} ( \res{0}{g} ( a , b ) , \res{1}{g} ( a , b ) )$ of the universe above the family $(\res{0}{g} ( a , b ) , \res{1}{g} ( a , b ))$, and so on.
\end{exa}


\subsection{Construction of A Higher-Order Universe Operator}\label{sec:high}
Following Palmgren \cite{palmgren1998}, we define the type $\rO$ of \textit{first-order operators} and the type $\Fam{\rO}$ of \textit{families of first-order operators} as
\[
\rO := \Fam{\rV} \to \Fam{\rV}, \qquad \qquad \Fam{\rO} := \Sigma_{(x : \rV)} (\rTV (x) \to \rO ) .
\]
For instance, let $f$ be of type $\rN_3 \to \rO$. Then, $(\wht{\rN_3}_{\rV} , f) : \Fam{\rO}$ is a family of first-order operators which consists of the operators $f \; 0_3 , f \; 1_3 , f \; 2_3 : \rO$. We note that $\rO$ and $\Fam{\rO}$ are just the lowest fragment of Palmgren's higher-order operators defined in \cite{palmgren1998}: he introduced the notion of $n$-order operator for an arbitrary natural number $n$.

We define the higher-order operator $\ru^{\mathbb{M}} : \Fam{\rO} \to \rO$ by using the reflection property of $\rV$. Typically, this operator takes a family of universe operators as an input, and returns a universe operator which provides a universe closed under all of the operators in this family. So an output of $\ru^{\mathbb{M}}$ can be its input again, and this process will be iterated transfinitely to construct the type-theoretic counterparts of inaccessible sets.
\begin{defi}[Operator $\ru^{\mathbb{M}}$, \href{https://github.com/takahashi-yt/czf-in-mahlo/blob/67c9dbf50bee67b73c1ea009929f4838458b8edc/src/ExternalMahlo.agda\#L101}{\texttt{Link}}]\label{def:uop}
  The higher-order operator $\ru^{\mathbb{M}} : \Fam{\rO} \to \rO$ is defined as follows. Let $(z , v) : \Fam{\rO}$ and $(x , y) : \Fam{\rV}$ be given. Then, by using $+ \el$ repeatedly, one can define
  \[
  h^{\mathbb{M}} : \Pi_{(w : \Fam{\rV})} \Bigl( \rN_1 + \rTV (x) + \rTV (z) + \Sigma_{(w' : \rTV (z))} \rTV (\rp_1 (v \: w' \: w)) \to \rV \Bigr)
  \]
  such that
  \begin{itemize}
  \item $h^{\mathbb{M}} \: w \: (\ri \ri \ri \: x_1) = x$ with $x_1 : \rN_1$,

  \item $h^{\mathbb{M}} \: w \: (\ri \ri \rj \: x_2) = y \: x_2$ with $x_2 : \rTV (x)$,

  \item $h^{\mathbb{M}} \: w \: (\ri \rj \: y_1) = \rp_1 (v \: y_1 \: w)$ with $y_1 : \rTV (z)$,

  \item $h^{\mathbb{M}} \: w \: (\rj (y_1 , z_1)) = \rp_2 (v \: y_1 \: w) \: z_1$ with $y_1 : \rTV (z)$ and $z_1 : \rTV (\rp_1 (v \: y_1 \: w))$.
    
  \end{itemize}
  Next, put $f^{\mathbb{M}} [z , v , x , y] : \rO$ as
  \[
  f^{\mathbb{M}} [z , v , x , y] := \lambda w . (\wht{\rN_1}_{\rV} \: \wht{+}_{\rV} \: x \: \wht{+}_{\rV} \: z \: \wht{+}_{\rV} \: \wht{\Sigma}_{\rV} (z , \lambda w' . \rp_1 (v \: w' \: w)) , \: h^{\mathbb{M}} \: w) .
  \]
  Then, we define $\ru^{\mathbb{M}} \: (z , v) \: (x , y) := (\wht{\rU}_{f^{\mathbb{M}} [z , v , x , y]} , \wht{\rT}_{f^{\mathbb{M}} [z , v , x , y]} )$. By applying $\Sigma \el$ twice, we have $\ru^{\mathbb{M}} : \Fam{\rO} \to \rO$.
\end{defi}

Roughly speaking, $h^{\mathbb{M}} \; w$ in the definition above is a family of small types in $\rV$ whose index set is
\[
\rN_1 + \rTV (x) + \rTV (z) + \Sigma_{(w' : \rTV (z))} \rTV (\rp_1 (v \: w' \: w)) .
\]
That is,
\begin{itemize}
\item $h^{\mathbb{M}} \: w \: (\ri \ri \ri \: x_1)$ is the small type $x : \rV$, where $x_1$ is of type $\rN_1$,

\item $h^{\mathbb{M}} \: w \: (\ri \ri \rj \: x_2)$ is the small type $y \: x_2 : \rV$ given by the family $y : \rTV (x) \to \rV$ of small types and the index $x_2 : \rTV (x)$,

\item $h^{\mathbb{M}} \: w \: (\ri \rj \: y_1)$ is the small type $\rp_1 (v \: y_1 \: w) : \rV$ given by the family $v : \rTV (z) \to \rO$ of first-order operators, the index $y_1 : \rTV (z)$ and $w : \Fam{\rV}$,

\item $h^{\mathbb{M}} \: w \: (\rj (y_1 , z_1))$ is the small type $\rp_2 (v \: y_1 \: w) \: z_1 : \rV$ given by the family $v : \rTV (z) \to \rO$ of first-order operators, the index $y_1 : \rTV (z)$, $w : \Fam{\rV}$ and $z_1 : \rTV (\rp_1 (v \: y_1 \: w))$.
    
\end{itemize}
In particular, the two small types $h^{\mathbb{M}} \: w \: (\ri \rj \: y_1)$ and $h^{\mathbb{M}} \: w \: (\rj (y_1 , z_1))$ are the outputs of the first-order operator $v \: y_1$ from the family $v : \rTV (z) \to \rO$.

For comparison, we show that super universes can be defined by using the operator $\ru^{\mathbb{M}}$ and this construction can be iterated. In \cite{palmgren1998} Palmgren formulated the family of theories for higher-order universe operators, and showed that it provides not only a construction of super universes, but also the iteration of this construction. The difference between our example and Palmgren's consists in the use of the Mahlo universe $\rV$: we defined the operator $\ru^{\mathbb{M}}$ by means of $\rV$, which is absent in Palmgren's theories for higher-order universe operators. Our example explains how to iterate the super-universe construction in $\MLM$.
\begin{exas}[Super Universe and Universes above Families of Small Types]\label{exa:opu}
  Take $g$ as defined in Example \ref{exa:suni}, and let $(a , b) : \Fam{\rV}$ be the case. Using the operator $\ru^{\mathbb{M}}$, we have the subuniverse of $\rV$ and its decoding function
  \[
  \ru^{\mathbb{M}} \: (\wht{\rN_1}_{\rV} , \lambda x .\rE_1 (x , g)) \: (a , b) = (\wht{\rU}_{g^{\ast}} , \wht{\rT}_{g^{\ast}} ) ,
  \]
  where we abbreviate $f^{\mathbb{M}} [\wht{\rN_1}_{\rV} , \lambda x .\rE_1 (x , g) , a , b]$ as $g^*$. This subuniverse has the codes for $a$ and $b \: c$ with $c : \rTV (a)$, as shown below.
  \begin{align*}
    \wht{\rT}_{g^*} (\res{1}{g^*} (\wht{\rN_0}_{g^*} , \lambda x . \rE_0 (x) , \ri \ri \ri \: 0_1 )) &= h^{\mathbb{M}} \: (\wht{\rN_0}_{\rV} , \lambda x . \wht{\rT}_{g^*} (\rE_0 (x))) \: (\ri \ri \ri \: 0_1) = a \\
    \wht{\rT}_{g^*} (\res{1}{g^*} (\wht{\rN_0}_{g^*} , \lambda x . \rE_0 (x) , \ri \ri \rj \: c )) &= h^{\mathbb{M}} \: (\wht{\rN_0}_{\rV} , \lambda x . \wht{\rT}_{g^*} (\rE_0 (x))) \: (\ri \ri \rj \: c) = b \: c
  \end{align*}
  Moreover, the subuniverse $\ru^{\mathbb{M}} \: (\wht{\rN_1}_{\rV} , \lambda x .\rE_1 (x , g)) \: (a , b)$ is closed under the universe operator $g$. Let $a' : \rU_{g^*}$ and $b' : \rT_{g^*} (a') \to \rU_{g^*}$ be given. We then have
  \begin{align*}
    \wht{\rT}_{g^*} (\res{1}{g^*} (a' , b' , \ri \rj \: 0_1)) &= h^{\mathbb{M}} \: (\wht{\rT}_{g^*}(a') , \lambda x . \wht{\rT}_{g^*} (b' \: x)) \: (\ri \rj \: 0_1) \\
    &= \rp_1 \Bigl( (\lambda y . \rE_1 (y , g)) \: 0_1 \: (\wht{\rT}_{g^*}(a') , \lambda x . \wht{\rT}_{g^*} (b' \: x)) \Bigr) \\
    &= \rp_1 \Bigl( g \: (\wht{\rT}_{g^*}(a') , \lambda x . \wht{\rT}_{g^*} (b' \: x)) \Bigr) = \wht{\rU}_{\lambda z . (\wht{\rT}_{g^*}(a') , \lambda x . \wht{\rT}_{g^*} (b' \: x))}\\
    \wht{\rT}_{g^*} (\res{1}{g^*} (a' , b' , \rj (0_1 , c))) &= h^{\mathbb{M}} \: (\wht{\rT}_{g^*}(a') , \lambda x . \wht{\rT}_{g^*} (b' \: x)) \: (\rj (0_1 , c)) \\
    &= \rp_2 \Bigl( (\lambda y . \rE_1 (y , g)) \: 0_1 \: (\wht{\rT}_{g^*}(a') , \lambda x . \wht{\rT}_{g^*} (b' \: x)) \Bigr) \: c \\
    &= \rp_2 \Bigl( g \: (\wht{\rT}_{g^*}(a') , \lambda x . \wht{\rT}_{g^*} (b' \: x)) \Bigr) \: c = \wht{\rT}_{\lambda z . (\wht{\rT}_{g^*}(a') , \lambda x . \wht{\rT}_{g^*} (b' \: x))} (c)
  \end{align*}
  Thus, $\res{1}{g^*} (a' , b' , \ri \rj \: 0_1)$ and $\lambda w' . \res{1}{g^*} (a' , b' , \rj (0_1 , w'))$ provide the codes in $\rU_{g^*}$ for the universe obtained by applying the universe operator $g$ to $\wht{\rT}_{g^*}(a')$ and $\lambda x . \wht{\rT}_{g^*} (b'\: x)$. This construction can be iterated as
  \[
  \res{1}{g^*} (\res{1}{g^*} (a' , b' , \ri \rj \: 0_1) , \lambda w' . \res{1}{g^*} (a' , b' , \rj (0_1 , w')) , \ri \rj \: 0_1) ,
  \]
  and so on.

  Note that $g^*$ is a first-order super-universe operator. So one can construct a universe closed under this super-universe operator by replacing $g$ with $g^*$ in the argument above, and this process can be iterated.
\end{exas}


\section{Inaccessible Sets in the Mahlo Universe of $\MLM$}\label{sec:large}
Below we first recall Aczel's constructive set theory $\CZF$ and then present Rathjen-Griffor-Palmgren's extension of $\CZF$ by inaccessible sets of all transfinite orders (Section~\ref{sec:czf}). Next, we prove our main theorem that inaccessible sets of $\CZF_{\pi}$ are interpretable in $\MLM$ with the accessibility predicate $\Acc$ (Section~\ref{sec:inacc}).


\subsection{$\CZF$ with Inaccessible Sets of All Transfinite Orders}\label{sec:czf}
The language of $\CZF$ is a usual first-order language with the equality $=$ and the membership relation $\in$ except that the restricted quantifiers $\forall x \in y$ and $\exists x \in y$ are primitive symbols. The symbol $\in$ is the only non-logical symbol in this language. The underlying logic of $\CZF$ is intuitionistic logic with equality, and the following set-theoretic axioms are postulated (\cf~\cite{AczelRathjen2001}).
\begin{description}
\item[Pairing]
  \[
  \forall a \forall b \exists c \forall x (x \in c \leftrightarrow x = a \lor x = b)
  \]

\item[Union]
  \[
  \forall a \exists b \forall x (x \in b \leftrightarrow \exists y \in a (x \in y))
  \]

\item[Infinity]
  \[
  \exists a (\exists x (x \in a) \land \forall x \in a \exists y \in a (x \in y))
  \]

\item[Extensionality]
  \[
  \forall a \forall b (\forall x (x \in a \leftrightarrow x \in b) \to a = b)
  \]

\item[Set Induction] For any formula $\varphi [x]$,
  \[
  \forall a (\forall x \in a \varphi [x] \to \varphi [a]) \to \forall a \varphi [a] .
  \]

\item[Restricted Separation] For any restricted formula $\varphi [x]$,
  \[
  \forall a \exists b \forall x (x \in b \leftrightarrow x \in a \land \varphi [x]) ,
  \]
  where a formula is \textit{restricted} if all of its quantifiers are restricted.

\item[Subset Collection] For any formula $\psi [x , y , u]$,
  \begin{align*}
    &\forall a \forall b \exists c \forall u \\
    &(\forall x \in a \exists y \in b \psi [x , y , u] \to \exists z \in c (\forall x \in a \exists y \in z \psi [x , y , u] \land \forall y \in z \exists x \in a \psi [x , y , u])) .
  \end{align*}

\item[Strong Collection] For any formula $\psi [x , y]$,
  \[
  \forall a (\forall x \in a \exists y \psi [x , y] \to \exists b (\forall x \in a \exists y \in b \psi [x , y] \land \forall y \in b \exists x \in a \psi [x , y])) .
  \]
\end{description}
Informally, one may suppose that the range of bound variables in $\CZF$ forms a collection of objects. We call these objects Aczel's \textit{iterative sets} or \textit{sets} shortly.

To state the axiom that there are inaccessible sets of all transfinite orders, we need to define the notion of regular set. For any set $a$, $a$ is \textit{transitive} if $\forall x \in a \forall y \in x (y \in a)$ holds.
\begin{defi}[Regular Sets]
  A set $a$ is regular if the following conditions are satisfied:
  \begin{enumerate}
  \item $a$ is inhabited, i.e., $\exists y (y \in a)$,

  \item $a$ is transitive,

  \item $\forall b \in a \forall R (R \subseteq b \times a \land \forall x \in b \exists y (\langle x , y \rangle \in R) \to$

    $\exists c \in a (\forall x \in b \exists y \in c (\langle x , y \rangle \in R) \land \forall y \in c \exists x \in b (\langle x , y \rangle \in R))$.
  \end{enumerate}
\end{defi}

The notion of inaccessible set is captured by \textit{set-inaccessibility}.
\begin{defi}
  A set $a$ is set-inaccessible if the following conditions are satisfied:
  \begin{enumerate}
  \item $a$ is regular,

  \item $\forall x \in a \exists y \in a (x \subseteq y \land y \text{ is regular})$,

  \item $a$ is a model of $\CZF$, i.e., the structure $\langle a , \in \upharpoonright (a \times a) \rangle$ is a model of $\CZF$.
  \end{enumerate}
\end{defi}

As in classical Zermelo-Fraenkel set theory, one can talk about \textit{classes} in $\CZF$. A class $A$ is said to be \textit{unbounded in} a class $B$ if $\forall x \in B \exists y \in A (x \in y \land y \in B)$ holds. Let $\TC{a}$ be the \textit{transitive closure} of a set $a$ with the following property (see \cite[Proposition 4.2]{AczelRathjen2001}):
\[
\TC{a} = a \cup \bigcup \{ \TC{x} \mid x \in a \} .
\]
Inaccessible sets of transfinite orders are defined along the hierarchy of Aczel's iterative sets.
\begin{defi}\label{def:ainacc}
  Let $a$ and $b$ be sets, then $b$ is $a$-set-inaccessible if $b$ is set-inaccessible and there is a family $(x_c)_{c \in \TC{a}}$ of sets such that for any $c \in \TC{a}$,
  \begin{itemize}
  \item $x_c$ is unbounded in $b$,

  \item $x_c$ consists of sets which are set-inaccessible,

  \item for any $y \in x_c$ and any $d \in \TC{c}$, $x_d$ is unbounded in $y$.
  \end{itemize}
\end{defi}
Informally speaking, a set $a$ represents a transfinite ``order'' in the hierarchy of iterative sets, so $a$-set-inaccessibles can be considered as inaccessible sets of order $a$. Note that $x_d$ in the third condition is well-defined because for any sets $a , c$ and $y$, if $c \in \TC{a}$ and $y \in \TC{c}$ hold then $y \in \TC{a}$ holds. This claim is proved by $\mathsf{TC}$-induction (see \cite[Proposition 4.3]{AczelRathjen2001}) on $a$: if $c \in a$ holds, then we have $y \in \TC{a}$ by $y \in \bigcup \{ \TC{x} \mid x \in a \}$. If $c \in \bigcup \{ \TC{x} \mid x \in a \}$ holds, then $c$ is an element of $\TC{x}$ for some $x \in \TC{a}$. By the hypothesis of $\mathsf{TC}$-induction, we have $y \in \TC{x}$ and so $y \in \TC{a}$ holds.

The first lemma below follows from the definition of $a$-set-inaccessibility, and the second lemma follows from the first (\cf~\cite[Lemma 2.11]{RathjenGrifforPalmgren1998}). For readers' convenience, we recall their proofs below.
\begin{lem}[$\CZF$]\label{lem:tcinacc}
  For any sets $a , b$ and $c$, if $a$ is $b$-set-inaccessible and $c \in \TC{b}$ holds, then $a$ is $c$-set-inaccessible.
\end{lem}
\begin{proof}
  By assumption, there is a family $(x_d)_{d \in \TC{b}}$ of sets witnessing the $b$-set-inaccessibility of $a$. To show its $c$-set-inaccessibility, we define a family $(y_d)_{d \in \TC{c}}$ of sets as follows: for any $d \in \TC{c}$, we have $d \in \TC{b}$, so we put $y_d := x_d$. It immediately follows that $y_d$ is unbounded in $a$ and that $y_d$ consists of sets which are set-inaccessible. Let $z \in y_d$ and $d' \in \TC{d}$ be the case. It follows that $z \in x_d$ holds and so $x_{d'}$, namely, $y_{d'}$ is unbounded in $z$.
\end{proof}

\begin{lem}[$\CZF$]\label{lem:setinacc}
  Let $a$ be set-inaccessible. Then, $a$ is $b$-set-inaccessible if and only if for all $c \in b$, the class of $c$-set-inaccessibles is unbounded in $a$.
\end{lem}
\begin{proof}
  Let $a$ be $b$-set-inaccessible, so there is a family $(x_c )_{c \in \TC{b}}$ witnessing the $b$-set-inaccessibility of $a$. We first show that for any $c \in \TC{b}$, $x_c$ consists of sets being $c$-set-inaccessible. Suppose that $c \in \TC{b}$ holds and $z$ is an element of $x_c$, then we verify that $z$ is $c$-set-inaccessible. Since $z$ is an element of $x_c$, $z$ is set-inaccessible. Next, we define a family $(y_d )_{d \in \TC{c}}$ of sets by putting $y_d := x_d$ for any $d \in \TC{c} \subseteq \TC{b}$. By definition, $y_d$ consists of sets being set-inaccessible. We have $z \in x_c$ and $d \in \TC{c}$, hence $x_d$ is unbounded in $z$ and so $y_d$ is too. For any $v \in y_d$ and any $d' \in \TC{d}$, we have $v \in x_d$ and so $x_{d'}$, \textit{i.e.}, $y_{d'}$ is unbounded in $v$. This concludes that $x_c$ consists of sets which are $c$-set-inaccessible for any $c \in \TC{b}$.

  If $c \in b$ holds then $c \in \TC{b}$ also holds. To show that the class of $c$-set-inaccessibles is unbounded in $a$, let $z \in a$ be the case. Since $x_c$ is unbounded in $a$, there is a set $v \in x_c$ such that $z \in v$ and $v \in a$ hold. This implies that the class of $c$-set-inaccessibles is unbounded in $a$, because $v$ is $c$-set-inaccessible.

  Conversely, suppose that for all $c \in b$, the class of $c$-set-inaccessibles is unbounded in $a$. It follows from Lemma~\ref{lem:tcinacc} that for all $d \in \TC{b}$, the class of $d$-set-inaccessibles is unbounded in $a$: if $d \in b$ holds then the class of $d$-set-inaccessibles is unbounded in $a$ by assumption. If $d \in \bigcup \{ \TC{x} \mid x \in b \}$ holds, then $d \in \TC{x}$ holds and the class of $x$-set-inaccessibles is unbounded in $a$ for some $x \in b$. Thus, for any $z \in a$, there is a set $v$ such that $v$ is $x$-set-inaccessible, $z$ is an element of $v$ and $v \in a$ holds. But $v$ is also $d$-set-inaccessible by Lemma~\ref{lem:tcinacc}.

  So we have $\forall d \in \TC{b} \forall z \in a \exists v \in a (z \in v \land v \text{ is $d$-set-inaccessible})$, that is,
  \begin{align*}
    &\forall d \in \TC{b} \forall z \in a \exists w \\
    &\quad \exists v \in a (w = \langle d , v , z \rangle \land z \in v \land v \text{ is $d$-set-inaccessible}) .
  \end{align*}
  Apply \textbf{Strong Collection} twice to this statement, and take the union of the set whose existence follows from the two applications of \textbf{Strong Collection}. We then have a set $S$ of triples $\langle d , v , z \rangle$ such that
  \begin{itemize}
  \item $z \in v \in a$ holds, $d \in \TC{b}$ holds and $v$ is $d$-set-inaccessible,

  \item for any $d \in \TC{b}$ and any $z \in a$, there is a set $v \in a$ with $\langle d , v , z \rangle \in S$.
    
  \end{itemize}
  Using \textbf{Restricted Separation}, we define
  \[
  S_d := \{ v \in a \mid \exists z \in a \langle d , v , z \rangle \in S \}
  \]
  for any $d \in \TC{b}$. Then, it follows that
  \begin{align*}
    &\forall d \in \TC{b} \forall v \in S_d \text{ there exists a family $(x_{d'})_{d' \in \TC{d}}$ of sets such that} \\
    &\quad \text{$(x_{d'})_{d' \in \TC{d}}$ witnesses the $d$-set-inaccessibility of $v$}.
  \end{align*}
  Again, apply \textbf{Strong Collection} twice to this statement, and take the union of the resulting set. We then obtain a set $\mathcal{F}$ of families of sets such that
  \begin{itemize}
  \item for any $d \in \TC{b}$ and any $v \in S_d$, there is a family $(x_{d'})_{d' \in \TC{d}} \in \mathcal{F}$ witnessing the $d$-set-inaccessibility of $v$,

  \item any $f \in \mathcal{F}$ is a family $(x_{d'})_{d' \in \TC{d}}$ witnessing the $d$-set-inaccessibility of $v$ for some $d \in \TC{b}$ and some $v \in S_d$.
    
  \end{itemize}
  We verify that $a$ is $b$-set-inaccessible. It follows from the assumption that $a$ is set-inaccessible. We define a family $(y_d)_{d \in \TC{b}}$ of sets: for any $d \in \TC{b}$,
  \[
  y_d := S_d \cup \bigcup \{ x_d \mid (x_{d_2})_{d_2 \in \TC{d_1}} \in \mathcal{F} \land d \in \TC{d_1} \text{ for some $d_1 \in \TC{b}$} \} .
  \]
  By definition, any element of $y_d$ is set-inaccessible. Moreover, $y_d$ is unbounded in $a$ because for any $z \in a$, there is a set $v \in a$ with $z \in v$ and $v \in S_d \subseteq y_d$.

  Finally, we show that for any $u \in y_d$ and any $d' \in \TC{d}$, $y_{d'}$ is unbounded in $u$. Let $s \in u$ be the case. If $u \in S_d$ holds, then there is a family $(x_{d'})_{d' \in \TC{d}} \in \mathcal{F}$ witnessing the $d$-set-inaccessibility of $u$. Since $x_{d'}$ is unbounded in $u$, we have $s \in s' \in u$ for some $s' \in x_{d'} \subseteq y_{d'}$, hence $y_{d'}$ is unbounded in $u$. If $u$ is an element of $x_{d}$ such that $d \in \TC{d_1}$ and $(x_{d_2})_{d_2 \in \TC{d_1}} \in \mathcal{F}$ hold for some $d_1 \in \TC{b}$, then $x_{d'}$ is unbounded in $u$ because the family $(x_{d_2})_{d_2 \in \TC{d_1}} \in \mathcal{F}$ witnesses the $d_1$-set-inaccessibility of some set (\cf~the last condition of Definition~\ref{def:ainacc}). So $s \in s' \in u$ holds for some $s' \in x_{d'} \subseteq y_{d'}$, hence $y_{d'}$ is unbounded in $u$ also in this case.
\end{proof}

The system $\CZF_{\pi}$ is obtained by adding the following axiom to $\CZF$.
\begin{description}
\item[Pi-Numbers]
  \[
  \forall x \forall a \exists b (x \in b \land b \text{ is $a$-set-inaccessible}).
  \]
\end{description}
That is, $\CZF_{\pi}$ is the extension of $\CZF$ by the axiom that for any sets $x$ and $a$, where $a$ represents a transfinite ``order'' in the hierarchy of iterative sets, there is a set $b$ which has $x$ as its element and is $a$-set-inaccessible. This axiom says that for any set $a$, the $a$-set-inaccessibles are unbounded.

As mentioned in Section~\ref{sec:intro}, Rathjen \cite{rathjen2003} showed that $\CZF_{\pi}$ is interpretable in the system $\CZFM$, which is the extension of $\CZF$ obtained by adding the axiom asserting the existence of one Mahlo set to $\CZF$. In the remainder of this subsection, we briefly recall Rathjen's result that $\CZF_{\pi}$ is interpretable in $\CZFM$, since this result is not explicitly stated in \cite{rathjen2003}.

Let $\mathbf{mv}(^a b)$ be the class of all subsets $R$ of $a \times b$ such that $\forall x \in a \exists y \in b (\langle x , y \rangle \in R)$ holds. The class $\mathbf{mv}(^a b)$ is thus the class of all multi-valued functions from $a$ to $b$. In \cite{rathjen2003}, the term ``set-inaccessibles'' means regular sets which are models of the fragment $\CZF^-$ of $\CZF$ obtained by omitting \textbf{Set Induction}. So we reformulate some definitions and lemmas in \cite{rathjen2003} as follows, but the lemmas are provable exactly in the same way as \cite{rathjen2003}.
\begin{defi}[Mahlo Sets]
  A set $M$ is Mahlo if $M$ is a regular set satisfying the following conditions:
  \begin{itemize}
  \item $M$ is a model of $\CZF^-$, and

  \item for any $R \in \mathbf{mv}(^M M)$, there is a regular model $I \in M$ of $\CZF^-$ such that $\forall x \in I \exists y \in I \langle x , y \rangle \in R$ holds.
    
  \end{itemize}
\end{defi}

\begin{lem}[$\CZF^-$]\label{lem:rone}
  Let $M$ be Mahlo. For any $R \in \mathbf{mv}(^M M)$ and any $a \in M$, there is a regular model $I \in M$ of $\CZF^-$ such that $a \in I$ and $\forall x \in I \exists y \in I \langle x , y \rangle \in R$ hold.
\end{lem}
\begin{proof}
  Suppose that $R \in \mathbf{mv}(^M M)$ holds and $a$ is an element of $M$. We define $S := \{ \langle x , \langle a , y \rangle \rangle \mid \langle x , y \rangle \in R \}$. Then $S \in \mathbf{mv}(^M M)$ holds because $M$ is closed under pairing, so there is a set $I \in M$ which is a regular model of $\CZF^-$, and satisfies $\forall x \in I \exists y \in I \langle x , y \rangle \in S$. Since $S$ is regular, $S$ is inhabited and so we have a set $x \in I$. We thus have $\langle x , \langle a , y \rangle \rangle \in S$ and $\langle a , y \rangle \in I$ for some $y$, hence $a \in I$ holds because $I$ is transitive. Similarly, for any $x \in I$, there is an ordered pair $\langle a , y \rangle \in I$ with $\langle x , y \rangle \in R$, and we have $y \in I$ because $I$ is transitive.
\end{proof}

\begin{lem}[$\CZF^-$]\label{lem:rtwo}
  Let $M$ be Mahlo. If $\forall x \in M \exists y \in M \: \varphi (x, y)$ holds, then there is a set $S \in \mathbf{mv}(^M M)$ satisfying $\forall x \forall y (\langle x , y \rangle \in S \to \varphi (x , y))$.
\end{lem}
\begin{proof}
  Define $\psi (x , z) := \exists y \in M (z = \langle x , y \rangle \land \varphi (x , y))$, then, by assumption, we have $\forall x \in M \exists z \in M \: \psi (x , z)$ because $M$ is closed under pairing. Using \textbf{Strong Collection}, there is a set $S$ with $\forall x \in M \exists z \in S \: \psi (x , z)$ and $\forall z \in S \exists x \in M \: \psi (x , z)$. The property $\forall x \in M \exists z \in S \: \psi (x , z)$ implies that $S \in \mathbf{mv}(^M M)$ holds, while the property $\forall z \in S \exists x \in M \: \psi (x , z)$ implies that $\forall x \forall y (\langle x , y \rangle \in S \to \varphi (x , y))$ holds.
\end{proof}

The following theorem is a straightforward adaptation of \cite[Theorem 4.10]{rathjen2003}, and this theorem implies that a Mahlo set $M$ models \textbf{Pi-Numbers}. Recall that a Mahlo set is a model of $\CZF^-$ by definition. Moreover, a Mahlo set also models \textbf{Set Induction}, because any Mahlo set is transitive and any transitive set models \textbf{Set Induction} in $\CZF$. It thus follows that $\CZF_{\pi}$ is interpretable in $\CZFM$: a Mahlo set $M$ in $\CZFM$ models all axioms of $\CZF_{\pi}$.
\begin{thm}[$\CZF$]
  Let $M$ be Mahlo. For any set $a \in M$, the class of $a$-set-inaccessibles is unbounded in $M$.
\end{thm}
\begin{proof}
  By induction on $a$. Let $x \in M$ be the case, and put a formula $\varphi (b , v)$ as
  \[
  \varphi (b , v) := \exists z \in M (v = \langle b , z \rangle \land x \in z \land \text{$z$ is $b$-set-inaccessible}) .
  \]
  The following argument shows that $\forall b \in a \exists v \varphi (b , v)$ holds. Let $b$ be a member of $a$, then we have
  \[
  b \in M \to \text{ the class of $b$-set-inaccessibles is unbounded in $M$}
  \]
  by IH. We have $b \in M$ by the transitivity of $M$, so the class of $b$-set-inaccessibles is unbounded in $M$. This implies that $\varphi (b , v)$ holds for some $v$.

  By \textbf{Strong Collection}, there is a set $R$ satisfying the following:
  \[
  \forall b \in a \exists v \in R \: \varphi (b , v) \land \forall v \in R \exists b \in a \: \varphi (b , v) .
  \]
  So $R$ is a subset of $a \times M$, and $\forall b \in a \exists z (\langle b , z \rangle \in R)$ holds. Since $M$ is regular, we have
  \[
  \forall b \in a \exists z \in y (\langle b , z \rangle \in R)
  \]
  for some $y \in M$. It then follows that for any $b \in a$, there exists a member $z$ of $y$ satisfying
  \begin{itemize}
  \item $\langle b , z \rangle = \langle b' , z' \rangle$ holds,

  \item $x \in z'$ holds, and

  \item $z'$ is $b'$-set-inaccessible
  \end{itemize}
  for some $b' \in a$ and $z' \in M$. But both $b = b'$ and $z = z'$ hold, so we have
  \[
  \forall x \in M \exists y \in M \forall b \in a \exists z \in y (x \in z \land \text{$z$ is $b$-set-inaccessible}).
  \]
  Lemmas \ref{lem:rtwo} and \ref{lem:rone} imply that for any $w \in M$, there exists a regular model $I \in M$ of $\CZF^-$ satisfying
  \[
  w \in I \land \forall x \in I \exists y \in I \forall b \in a \exists z \in y (x \in z \land \text{$z$ is $b$-set-inaccessible}).
  \]
  By the transitivity of $I$, we have
  \[
  \forall x \in I \forall b \in a \exists z \in I (x \in z \land \text{$z$ is $b$-set-inaccessible}) ,
  \]
  so the class of $b$-set-inaccessibles is unbounded in $I$ for any $b \in a$. By the transitivity of $I$ again, $I$ satisfies \textbf{Set Induction}, hence $I$ is also a model of $\CZF$. Moreover, one can see by definition that any $b$-set-inaccessible is regular, so the class of regular sets is unbounded in $I$. Therefore, $I$ is set-inaccessible.

  It follows by Lemma \ref{lem:setinacc} that $I$ is $a$-set-inaccessible, therefore $a$-set-inaccessibles are unbounded in $M$.
\end{proof}


\subsection{Construction of Inaccessible Sets in Mahlo Universe}\label{sec:inacc}
To show that $\CZF_{\pi}$ is interpretable in $\MLM$ with the accessibility predicate $\Acc$, we recall Aczel's interpretation of $\CZF$ in $\MLTT$.

First of all, we define the type-theoretic counterpart of Aczel's iterative sets. Let a $\rW$-type $\rW_{(x : A)} B$ be given. It is straightforward to define the functions $\idx : \rW_{(x : A)} B \to A$ and $\pred : \Pi_{(y : \rW_{(x : A)} B)} B [\idx \; y / x] \to \rW_{(x : A)} B$ satisfying
\[
\idx \; \wcons (a,b) = a , \qquad \pred \; \wcons (a,b) = b .
\]
The $\rW$-type $\bbV$, which is the type of iterative sets in Aczel's interpretation of $\CZF$, is defined as
\[
\bbV := \rW_{(x : \rV)} \rTV (x).
\]
We call terms of type $\bbV$ in some contexts \textit{sets}, and denote them by lowercase Greek letters such as $\alpha , \beta , \gamma$ possibly with suffixes. For any $\alpha : \bbV$, we write $\rTV (\idx \; \alpha)$ as $\ov{\alpha}$, following \cite{RathjenGrifforPalmgren1998}. Moreover, we adopt the abbreviations below: for any term $F$ such that $\Gamma , \beta : \bbV \vdash F \; \type$ holds for some context $\Gamma$,
\[
\forall \beta \in \alpha F := \Pi_{(x : \ov{\alpha})} F[\pred \: \alpha \: x / \beta] , \qquad \exists \beta \in \alpha F := \Sigma_{(x : \ov{\alpha})} F[\pred \: \alpha \: x / \beta] .
\]
Note that the elimination rule for $\bbV = \rW_{(x : \rV)} \rTV (x)$ provides the transfinite induction principle on $\alpha : \bbV$. That is, we have a term
\[
\TI_F : \Pi_{(\alpha : \bbV)}(\forall \gamma \in \alpha F[\gamma / \beta ] \to F[\alpha / \beta ]) \to \Pi_{(\alpha : \bbV)} F[\alpha / \beta] ,
\]
where $\Gamma , \beta : \bbV \vdash F \; \type$ holds for some context $\Gamma$.

Next, we define the interpretation of the equality $a = b$ and the membership relation $a \in b$ in $\CZF$ as an equivalence relation $\doteq$ and a binary relation $\in$ on $\bbV$, respectively. The relation $\doteq$ on $\bbV$ is defined by applying the elimination rule for $\rW$-types twice:
\[
\wcons (a , f) \doteq \wcons (b , g) := (\Pi_{(x : \rTV (a))}\Sigma_{(y : \rTV (b))} f \; x \doteq g \; y) \times (\Pi_{(y : \rTV (b))}\Sigma_{(x : \rTV (a))} f \; x \doteq g \; y) .
\]
For any $\alpha , \beta : \bbV$, we define $\alpha \in \beta$ as $\alpha \in \beta := \exists \gamma \in \beta (\alpha \doteq \gamma )$. Below we use the following facts about these relations $\doteq$ and $\in$ without mentioning them. Note that by Lemma~\ref{lem:eta}.(1) below we have $\alpha \doteq \wcons (\idx \: \alpha , \pred \: \alpha)$ for any $\alpha : \bbV$.
\begin{lem}[\href{https://github.com/takahashi-yt/czf-in-mahlo/blob/67c9dbf50bee67b73c1ea009929f4838458b8edc/src/CZFBasics.agda}{\texttt{Link}}]\label{lem:eta}
  The following holds:
  \begin{enumerate}
  \item For any $\alpha : \bbV$, we have $\alpha =_{\bbV} \wcons (\idx \: \alpha , \pred \: \alpha)$.

  \item For any $\alpha , \beta : \bbV$, if $\alpha \doteq \beta$ holds then $\Pi_{(\gamma : \bbV)}(\gamma \in \alpha \to \gamma \in \beta)$ holds.

  \item For any $a : \rV$ and $f , g : \rTV a \to \bbV$, if $\Pi_{(x : \rTV a)} f \: x \doteq g \: x$ holds then $\wcons (a , f) \doteq \wcons (a , g)$ holds.

  \item For any $\alpha , \beta : \bbV$, if both $\Pi_{(\gamma : \bbV)}(\gamma \in \alpha \to \gamma \in \beta)$ and $\Pi_{(\gamma : \bbV)}(\gamma \in \beta \to \gamma \in \alpha)$ holds then $\alpha \doteq \beta$ holds.
    
  \end{enumerate}
\end{lem}
\begin{proof}
  The statements (1) and (2) are provable by the induction principle of $\bbV$. The statement (3) follows straightforwardly, and (4) holds because the axiom \textbf{Extensionality} is interpretable in $\MLTT$.
\end{proof}

We interpret all formulas of $\CZF$ in $\MLTT$ by using these relations $\doteq$ and $\in$ on $\bbV$. The formulas $x = y$ and $x \in y$ in the language of $\CZF$ are interpreted as the relations $x \doteq y$ and $x \in y$ on $\bbV$, respectively. Compound formulas of $\CZF$ are interpreted as follows:
\begin{align*}
  (\varphi \land \psi )^{\ast} &:= \varphi^{\ast} \times \psi^{\ast} & (\varphi \lor \psi )^{\ast} &:= \varphi^{\ast} + \psi^{\ast} \\
  (\varphi \to \psi )^{\ast} &:= \varphi^{\ast} \to \psi^{\ast} & \bot^{\ast} &:= \rN_0 \\
  (\forall x \in y \varphi )^{\ast} &:= \forall x \in y \varphi^{\ast} & (\exists x \in y \varphi )^{\ast} &:= \exists x \in y \varphi^{\ast} \\
  (\forall x \varphi )^{\ast} &:= \Pi_{(\alpha : \bbV)} \varphi^{\ast} & (\exists x \varphi )^{\ast} &:= \Sigma_{(\alpha : \bbV)} \varphi^{\ast}
\end{align*}
It is then straightforward to show by induction on the length of a given $\CZF$-formula $\varphi$ that if a list $x_1 , \ldots , x_n$ of variables contains all free variables in $\varphi$ then the judgement $x_1 : \bbV , \ldots , x_n : \bbV \vdash \varphi^{\ast} \; \type$ is derivable in $\MLTT$. Moreover, as suggested in \cite[Remarks 3.5]{aczel1982}, the following useful lemma is provable by induction on the length of a given $\CZF$-formula.
\begin{lem}\label{lem:czfext}
  Let $\varphi [x_1 , \ldots , x_n , x]$ be an arbitrary formula of $\CZF$ such that $x_1 , \ldots , x_n$ and $x$ contain all free variable occurring in $\varphi$. Then, the statement
  \[
  \Pi_{(x_1 : \bbV)} \cdots \Pi_{(x_n : \bbV)} \Pi_{(y : \bbV)} \Pi_{(z : \bbV)} (y \doteq z \to \varphi^{\ast} [x_1 , \ldots , x_n , y] \to \varphi^{\ast} [x_1 , \ldots , x_n , z])
  \]
  is derivable in $\MLTT$.
\end{lem}

In $\MLTT$, one is able to define the corresponding operators to the set-theoretic axioms of $\CZF$ (see \cite{aczel1978,aczel1982}). For instance, the union operator and the binary union operator for elements of $\bbV$ can be defined: let $A \leftrightarrow B$ be an abbreviation of $(A \to B) \times (B \to A)$, then the type-theoretic axiom of choice in $\MLTT$ provides the two terms $\bigcup : \bbV \to \bbV$ and $\cup : \bbV \to \bbV \to \bbV$ satisfying
\begin{align*}
  &\Pi_{(\alpha : \bbV )} \Pi_{(\gamma : \bbV)} (\gamma \in \bigcup \alpha \leftrightarrow \exists \delta \in \alpha (\gamma \in \delta )) ,\\
  &\Pi_{(\alpha_1 : \bbV)} \Pi_{(\alpha_2 : \bbV)} \Pi_{(\gamma : \bbV)} (\gamma \in \alpha_1 \cup \alpha_2 \leftrightarrow (\gamma \in \alpha_1 ) + (\gamma \in \alpha_2 ) ) .
\end{align*}
The term $\bigcup$ enables to interpret the axiom \textbf{Union} in $\MLTT$, that is, one can show that the interpretation of the axiom \textbf{Union} is derivable in $\MLTT$. In a similar manner, one can verify that all axioms of $\CZF$ are interpretable in $\MLTT$ (see \cite{aczel1978,aczel1982} for details).

For our purpose, the axiom \textbf{Pi-Numbers} remains to be interpreted in $\MLM$. Below we first define the transitive closure function $\mathsf{tc}$ on $\bbV$ to provide an interpretation of the set-theoretic transitive closure. Next, we define the accessibility predicate $\Acc$ on $\bbV$ and construct the type-theoretic counterparts $\itebV{\alpha}{t}{a}{f}$ of inaccessible sets by iterating the higher-order universe operator $\ru^{\rV}$ along the accessibility given by $\Acc$. Finally, it is shown that inaccessible sets of all transfinite orders in $\CZF$ are indeed interpreted by these counterparts.

Following \cite{RathjenGrifforPalmgren1998}, we define the transitive closure function on $\bbV$ as follows.
\begin{defi}[Transitive Closure, \href{https://github.com/takahashi-yt/czf-in-mahlo/blob/67c9dbf50bee67b73c1ea009929f4838458b8edc/src/CZFAxioms.agda\#L916}{\texttt{Link}}]
  The transitive closure function $\mathsf{tc} : \bbV \to \bbV$ is defined as
  \[
  \mathsf{tc} (\wcons (a , f)) := \wcons (a , f) \cup (\bigcup \wcons (a , \lambda x . \mathsf{tc} (f \; x) )).
  \]
  We put $\tc{\alpha} := \mathsf{tc} (\alpha )$ as in \cite{RathjenGrifforPalmgren1998}.
\end{defi}

Note that, by definition, $\ov{\tc{\wcons (a , f)}}$ is definitionally equal to $\rTV (a) + \Sigma_{(z : \rTV (a))} \ov{\tc{(f \; z)}}$. In addition, we define the accessibility predicate on $\bbV$ as an inductive predicate. This predicate is used to iterate the higher-order operator $\ru^{\rV}$ along a given $\alpha : \bbV$ (\cf~Definition \ref{def:phi}).
\begin{defi}[Accessibility Predicate, \href{https://github.com/takahashi-yt/czf-in-mahlo/blob/67c9dbf50bee67b73c1ea009929f4838458b8edc/src/TypeHierarchy.agda\#L19}{\texttt{Link}}]\label{def:acc}
  The predicate $\Acc$ on $\bbV$ is defined inductively: the formation, introduction, elimination and computation rules are as follows.
  \[
  \infer[\fm]{\Gamma \vdash \Acc \: \alpha \; \type}{
    \Gamma \vdash \alpha : \bbV
  }
  \qquad
  \infer[\intro]{\Gamma \vdash \prog \: f : \Acc \: \alpha}{
    \deduce{\Gamma \vdash f : \Pi_{(x : \ov{\tc{\alpha}})} \Acc \: (\pred \: \tc{\alpha} \: x)}{\Gamma \vdash \alpha : \bbV}
  }
  \]
  \[
  \infer[\el]{\Gamma \vdash \rE_{\Acc}(\alpha , t , g) : C [\alpha / \beta , t / x]}{
    \deduce{\Gamma \vdash g : \Pi_{(\gamma : \bbV)}\Pi_{(f : \Pi_{(y : \ov{\tc{\gamma}})} \Acc \: (\pred \: \tc{\gamma} \: y))}((\Pi_{(y : \ov{\tc{\gamma}})} C[\pred \: \tc{\gamma} \: y / \beta , f \: y / x]) \to C[\gamma / \beta , \prog \: f / x])}{
      \deduce{\Gamma , \beta : \bbV , x : \Acc \: \beta \vdash C \; \type}{
        \deduce{\Gamma \vdash t : \Acc \: \alpha}{\Gamma \vdash \alpha : \bbV}
      }
    }
  }
  \]
  \[
  \infer[\eq]{\Gamma \vdash \rE_{\Acc}(\alpha , \prog \: f , g) = g \: \alpha \: f \: (\lambda y . \rE_{\Acc} (\pred \: \tc{\alpha} \: y , f \: y , g)) : C [\alpha / \beta , \prog \: f / x]}{
    \deduce{\Gamma \vdash g : \Pi_{(\gamma : \bbV)}\Pi_{(f : \Pi_{(y : \ov{\tc{\gamma}})} \Acc \: (\pred \: \tc{\gamma} \: y))}((\Pi_{(y : \ov{\tc{\gamma}})} C[\pred \: \tc{\gamma} \: y / \beta , f \: y / x]) \to C[\gamma / \beta , \prog \: f / x])}{
      \deduce{\Gamma , \beta : \bbV , x : \Acc \: \beta \vdash C \; \type}{
        \deduce{\Gamma \vdash f : \Pi_{(y : \ov{\tc{\alpha}})} \Acc \: (\pred \: \tc{\alpha} \: y)}{\Gamma \vdash \alpha : \bbV}
      }
    }
  }
  \]
\end{defi}

Informally speaking, the family $\bigcup_{x : \ov{\tc{\alpha}}} \pred \: \tc{\alpha} \: x$ of sets exhausts all sets which are needed to construct the set $\alpha$, and the introduction rule of $\Acc$ means that $\alpha$ is accessible if any of such sets is accessible. The elimination rule of $\Acc$ enables to provide an inductive argument for $C[\gamma / \beta , \prog \: f / x]$ with the induction hypothesis that $C [\pred \: \tc{\gamma} \: y / \beta , f \: y / x]$ holds for any $y : \ov{\tc{\gamma}}$.

We denote by $\MLMacc$ the system obtained by extending $\MLM$ with the rules for $\Acc$. In the remainder of this section we work in $\MLMacc$. The induction principle on $\Acc$ provides the iterator of the operator $\ru^{\rV}$ below.
\begin{defi}[Iterator of $\ru^{\rV}$ along Accessibility, \href{https://github.com/takahashi-yt/czf-in-mahlo/blob/67c9dbf50bee67b73c1ea009929f4838458b8edc/src/TypeHierarchy.agda\#L41}{\texttt{Link}}]\label{def:phi}
  We define $\Phi : \Pi_{(\alpha : \bbV)} (\Acc \: \alpha \to \rO)$ by induction on $t : \Acc \: \alpha$:
  \[
  \Phi \: \alpha \: (\prog \: f) \: (c , h) := \ru^{\rV} (\idx \: \tc{\alpha}, \lambda x . \Phi \: (\pred \: \tc{\alpha} \: x) \: (f \: x)) \: (c , h) .
  \]
  We abbreviate $\Phi \: \alpha \: t \: (c , h)$ as $\iteHh{\alpha}{t}{c}{h}$.
\end{defi}
\pagebreak[5]

Here the iterator $\Phi$ takes a set $\alpha$ and a proof $\prog \: f$ of its accessibility with a family $(c , h) : \Fam{\rV}$ and returns the subuniverse $\rU$ of $\rV$ such that
\begin{itemize}
\item $\rU$ is closed under the universe operator $\Phi \: (\pred \: \tc{\alpha} \: x) \: (f \: x)$ for any $x : \ov{\tc{\alpha}}$, which is obtained by the induction hypothesis, and

\item $\rU$ includes the codes of $c$ and $h$.
\end{itemize}
The hierarchy of such subuniverses is defined by applying the left and right projection to the output of the iterator $\Phi$.
\begin{defi}[Transfinite Hierarchy of Subuniverses, \href{https://github.com/takahashi-yt/czf-in-mahlo/blob/67c9dbf50bee67b73c1ea009929f4838458b8edc/src/TypeHierarchy.agda\#L50}{\texttt{Link}}]\label{def:subuni}
  Let $\alpha : \bbV$, $t : \Acc \: \alpha$ and $(a , f) : \Fam{\rV}$ be given. We define the transfinite hierarchy of subuniverses with their decoding functions $\iteTh{\alpha}{t}{a}{f}$ along $\Acc$.
  \begin{align*}
    \iteMh{\alpha}{t}{a}{f} &:= \rp_1 (\iteHh{\alpha}{t}{a}{f}) : \rV & \iteM{\alpha}{t}{a}{f} &:= \rTV (\iteMh{\alpha}{t}{a}{f}) \\
    \iteTh{\alpha}{t}{a}{f} &:= \rp_2 (\iteHh{\alpha}{t}{a}{f}) : \iteM{\alpha}{t}{a}{f} \to \rV & \iteT{\alpha}{t}{a}{f} &:= \lambda x . \rTV (\iteTh{\alpha}{t}{a}{f} (x))
  \end{align*}
\end{defi}

Each of subuniverses $\iteM{\alpha}{t}{a}{f}$ provides the type $\iteV{\alpha}{t}{a}{f}$ of iterative sets on it, as we will define below. The transfinite hierarchy of iterative sets consists of the sets each of which is obtained by embedding the sets in $\iteV{\alpha}{t}{a}{f}$ into $\bbV$ via the mapping $\iteh{\alpha}{t}{a}{f}$ defined below. One can find the Agda code for the definition of this hierarchy of iterative sets in \href{https://github.com/takahashi-yt/czf-in-mahlo/blob/67c9dbf50bee67b73c1ea009929f4838458b8edc/src/IterativeSetHierarchy.agda}{\texttt{IterativeSetHierarchy.agda}} of \cite{takahashi2025}, while the definition of the type $\iteV{\alpha}{t}{a}{f}$ is formalised in \href{https://github.com/takahashi-yt/czf-in-mahlo/blob/67c9dbf50bee67b73c1ea009929f4838458b8edc/src/TypeHierarchy.agda}{\texttt{TypeHierarchy.agda}}.
\begin{defi}[Transfinite Hierarchy of Iterative Sets]\label{def:trans}
  The transfinite hierarchy of iterative sets is defined via the type $\iteV{\alpha}{t}{a}{f}$ and the mapping $\iteh{\alpha}{t}{a}{f} : \iteV{\alpha}{t}{a}{f} \to \bbV$:
  \begin{align*}
    \iteVh{\alpha}{t}{a}{f} &:= \wht{\rW}_{\rV} ( \iteMh{\alpha}{t}{a}{f} , \iteTh{\alpha}{t}{a}{f} ) : \rV \\
    \iteV{\alpha}{t}{a}{f} &:= \rTV (\iteVh{\alpha}{t}{a}{f}) \\
    \iteh{\alpha}{t}{a}{f} \wcons (b , g) &:= \wcons (\iteTh{\alpha}{t}{a}{f} (b) , \lambda x . \iteh{\alpha}{t}{a}{f} (g \: x) ) \\
    \itebV{\alpha}{t}{a}{f} &:= \wcons (\iteVh{\alpha}{t}{a}{f} , \iteh{\alpha}{t}{a}{f}) : \bbV
  \end{align*}
\end{defi}
Lemma \ref{lem:lastone} below shows that $\itebV{\alpha}{t}{a}{f}$ is $\alpha$-set-inaccessible: it is the hierarchy defined above that validates the axiom \textbf{Pi-Numbers} in $\MLMacc$. The set-theoretic terms ``transitive'', ``regular'', ``$\alpha$-set-inaccessible'' and ``unbounded'' below are in fact their translations in $\MLTT$.

In a way similar to define the relations $\doteq$ and $\in$ on $\bbV$, the equivalence relation $\doteq^{\alpha , t}_{(a , f)}$ and the membership relation $\in^{\alpha , t}_{(a , f)}$ on $\iteV{\alpha}{t}{a}{f}$ can be defined, as remarked in \cite{RathjenGrifforPalmgren1998}. One can prove the following lemma for these relations. The Agda code for the proofs of the statements (1), (3) and (4) can be found in \href{https://github.com/takahashi-yt/czf-in-mahlo/blob/67c9dbf50bee67b73c1ea009929f4838458b8edc/src/TypeHierarchy.agda}{\texttt{TypeHierarchy.agda}} of \cite{takahashi2025}. On the other hand, the code for the proofs of statements (2) and (5) can be found in \href{https://github.com/takahashi-yt/czf-in-mahlo/blob/67c9dbf50bee67b73c1ea009929f4838458b8edc/src/IterativeSetHierarchy.agda}{\texttt{IterativeSetHierarchy.agda}} and \href{https://github.com/takahashi-yt/czf-in-mahlo/blob/67c9dbf50bee67b73c1ea009929f4838458b8edc/src/Inaccessibles/Basics.agda}{\texttt{Inaccessibles/Basics.agda}}, respectively. Note that the Agda code for the statement (5) depends on a postulate on which we will comment at the end of this section.
\begin{lem}\label{lem:lemsone}
  The following holds:
  \begin{enumerate}
  \item For any $x : \iteV{\alpha}{t}{a}{f}$, we have $x \doteq^{\alpha , t}_{(a , f)} \wcons (\idx \: x , \pred \: x)$.
    
  \item $\itebV{\alpha}{t}{a}{f}$ is transitive.

  \item For any $x , y : \iteV{\alpha}{t}{a}{f}$, $x \in^{\alpha , t}_{(a , f)} y$ holds if and only if $\iteh{\alpha}{t}{a}{f} x \in \iteh{\alpha}{t}{a}{f} y$ holds.

  \item For any $x , y : \iteV{\alpha}{t}{a}{f}$, $x \doteq^{\alpha , t}_{(a , f)} y$ holds if and only if $\iteh{\alpha}{t}{a}{f} x \doteq \iteh{\alpha}{t}{a}{f} y$ holds.

  \item $\itebV{\alpha}{t}{a}{f}$ is set-inaccessible.
    
  \end{enumerate}
\end{lem}
\begin{proof}
  The statement (1) is proved by induction on $\iteV{\alpha}{t}{a}{f}$. For the proofs of (2), (3) and (4), see \cite[Lemma 4.7]{RathjenGrifforPalmgren1998}; for the proof of (5), see \cite[Corollary 4.8]{RathjenGrifforPalmgren1998}. Note that the definition of transfinite hierarchy of iterative sets in \cite{RathjenGrifforPalmgren1998} does not use the accessibility predicate $\Acc$. Moreover, the proofs in \cite{RathjenGrifforPalmgren1998} do not assume much more than that each type $\iteV{\alpha}{t}{a}{f}$ is an instance of Aczel's type of iterative sets (\cf~\cite[\S~1.9]{aczel1982} and \cite[\S~1]{aczel1986}).
\end{proof}

By a standard argument using the induction principle of identity types, we have the so-called indiscernibility of identicals, which is also called the transport lemma: for any $A , B$ with $\Gamma , z : A \vdash B \; \type$,
\[
\mathsf{transport}^B : \Pi_{(x : A)} \Pi_{(y : A)} (x =_A y \to B [x / z] \to B [y / z]) .
\]
For the sake of brevity, we often treat the first two arguments of $\mathsf{transport}^B$ as implicit.
\begin{lem}[\href{https://github.com/takahashi-yt/czf-in-mahlo/blob/67c9dbf50bee67b73c1ea009929f4838458b8edc/src/CZFBasics.agda\#L121}{\texttt{Link}}]\label{lem:transp}
  For any $a , b : \rV$, if $p : a =_{\rV} b$ holds then for any $g : \rTV (b) \to \bbV$,
  \[
  \wcons (a , \lambda x . g \: ( \mathsf{transport}^{\rTV (z)} \: a \: b \: p \: x)) =_{\bbV} \wcons (b , g)
  \]
  holds.
\end{lem}
\begin{proof}
  One can prove the lemma by the induction principle of identity types: it suffices to verify that the claim
  \[
  \wcons (a , \lambda x . g \: ( \mathsf{transport}^{\rTV (z)} \: a \: a \: \refl (a) \: x)) =_{\bbV} \wcons (a , g)
  \]
  holds. Since $\mathsf{transport}^{\rTV (z)} \: a \: a \: \refl (a)$ is definitionally equal to the identity function on $\rTV (a)$, the claim holds by the reflexivity of identity types and the $\eta$-rule for $\Pi$-types.
\end{proof}

We prove one more auxiliary lemma below.
\begin{lem}[\href{https://github.com/takahashi-yt/czf-in-mahlo/blob/67c9dbf50bee67b73c1ea009929f4838458b8edc/src/TypeHierarchy.agda\#L56}{\texttt{Link}}]\label{lem:code}
  For any $\alpha : \bbV$, $t : \Acc \: \alpha$, $a : \rV$, $g : \rTV (a) \to \rV$ and $x : \rTV (a)$, there is a term $c$ of type $\iteM{\alpha}{t}{a}{g}$ such that $\iteTh{\alpha}{t}{a}{g} (c) =_{\rV}  g \: x$ holds.
\end{lem}
\begin{proof}
  We prove the lemma by transfinite induction on $t : \Acc \; \alpha$, so let $t$ be $\prog \: f'$. We have
  \begin{align*}
  \iteM{\alpha}{\prog \: f'}{a}{g} &= \rTV (\rp_1 (\Phi \: \alpha \: (\prog \: f') \: (a , g) )) \\
    &= \rTV (\rp_1 (\ru^{\rV} (\idx \: \tc{\alpha}, \lambda y . \Phi \: (\pred \: \tc{\alpha} \: y) \: (f' \: y)) \: (a , g))) \\
    &= \rU_{f^{\rV} [\idx \: \tc{\alpha} , \lambda y . \Phi \: (\pred \: \tc{\alpha} \: y) \: (f' \: y) , a , g]} .
  \end{align*}
  Put $f^{\ast} := f^{\rV} [\idx \: \tc{\alpha} , \lambda y . \Phi \: (\pred \: \tc{\alpha} \: y) \: (f' \: y) , a , g]$. Note that we have $\rU_{f^{\ast}} = \iteM{\alpha}{\prog \: f'}{a}{g}$ and $\wht{\rT}_{f^{\ast}} = \iteTh{\alpha}{\prog \: f'}{a}{g}$. It then follows that
  \[
  \wht{\rT}_{f^{\ast}} (\res{1}{f^{\ast}} (\wht{\rN_0}_{f^{\ast}} , \lambda v . \rE_0 v , \ri \ri \rj \: x)) = \rp_2 (f^{\ast} (\wht{\rT}_{f^{\ast}} (\wht{\rN_0}_{f^{\ast}}) , \lambda v . \wht{\rT}_{f^{\ast}} (\rE_0 v) , \ri \ri \rj \: x)) = g \: x
  \]
  holds. So the term $\res{1}{f^{\ast}} (\wht{\rN_0}_{f^{\ast}} , \lambda v . \rE_0 v , \ri \ri \rj \: x) : \rU_{f^{\ast}}$ satisfies the claim.
\end{proof}

The following is the main lemma in our proof. Informally, this lemma says that $\itebV{\alpha}{\prog \: f'}{a}{f}$ is unbounded in $(\pred \: \tc{\alpha} \: u)$-set-inaccessibles for any $u : \ov{\tc{\alpha}}$.
\begin{lem}[Main Lemma, \href{https://github.com/takahashi-yt/czf-in-mahlo/blob/67c9dbf50bee67b73c1ea009929f4838458b8edc/src/PiNumbersAxiom.agda\#L32}{\texttt{Link}}]\label{lem:main}
  If $\xi \in \itebV{\alpha}{\prog \: f'}{a}{f}$ holds, then for any $u : \ov{\tc{\alpha}}$, there are terms $b : \iteM{\alpha}{\prog \: f'}{a}{f}$ and $g : \iteT{\alpha}{\prog \: f'}{a}{f} (b) \to \iteM{\alpha}{\prog \: f'}{a}{f}$ such that
  \[
  \xi \in \itebV{\pred \: \tc{\alpha} \: u}{f' \: u}{b^{\ast}}{g^{\ast}} \in \itebV{\alpha}{\prog \: f'}{a}{f}
  \]
  holds with $b^{\ast} := \iteTh{\alpha}{\prog \: f'}{a}{f} (b)$ and $g^{\ast} := \lambda x . \iteTh{\alpha}{\prog \: f'}{a}{f}(g \: x)$.
\end{lem}
\begin{proof}
  Below we abbreviate $\prog \: f'$ as $t$. Since $\itebV{\alpha}{t}{a}{f}$ is set-inaccessible by Lemma \ref{lem:lemsone}.(5), the singleton $\{ \xi \}$ also belongs to $\itebV{\alpha}{t}{a}{f}$. Moreover, there is a regular set $\beta \in \itebV{\alpha}{t}{a}{f}$ such that $\xi \in \beta$ holds and $\beta$ is transitive. It follows that $\beta \doteq \iteh{\alpha}{t}{a}{f} x$ holds for some $x : \iteV{\alpha}{t}{a}{f}$. We have $x \doteq^{\alpha , t}_{(a , f)} \wcons (\idx \: x , \pred \: x)$ by Lemma \ref{lem:lemsone}.(1), so
  \[
  \beta \doteq \iteh{\alpha}{t}{a}{f} x \doteq \iteh{\alpha}{t}{a}{f} \wcons (\idx \: x , \pred \: x) = \wcons (\iteTh{\alpha}{t}{a}{f}(\idx \: x) , \lambda y . \iteh{\alpha}{t}{a}{f} (\pred \: x \: y) )
  \]
  holds by Lemma \ref{lem:lemsone}.(4).

  Put
  \begin{itemize}
  \item $b := \idx \: x : \iteM{\alpha}{t}{a}{f}$ and $g := \lambda z . \idx \: (\pred \: x \: z) : \iteT{\alpha}{t}{a}{f} (b) \to \iteM{\alpha}{t}{a}{f}$, where $\iteT{\alpha}{t}{a}{f} (b) = \rTV (\iteTh{\alpha}{t}{a}{f} (b))$ holds by definition,

  \item $b^{\ast} := \iteTh{\alpha}{t}{a}{f} (b) : \rV$ and $g^{\ast} := \lambda w . \iteTh{\alpha}{t}{a}{f}(g \: w) : \iteT{\alpha}{t}{a}{f} (b) \to \rV$,

  \item $\wht{\mathcal{U}} := \iteMh{\pred \: \tc{\alpha} \: u}{f' \: u}{b^{\ast}}{g^{\ast}}$.
  \end{itemize}
  
  We first show that $\iteM{\alpha}{t}{a}{f}$ has a code for $\wht{\mathcal{U}}$. We have by definition
  \begin{align*}
  \iteM{\alpha}{t}{a}{f} &= \rTV (\rp_1 (\Phi^{\alpha , t}_{(a , f)})) \\
    &= \rTV (\rp_1 (\ru^{\rV}(\idx \: \tc{\alpha} , \lambda y . \Phi \: (\pred \: \tc{\alpha} \: y) \: (f' \: y) ) (a , f) )) \\
    &= \rTV (\wht{\rU}_{f^{\rV} [\idx \: \tc{\alpha} , \lambda y . \Phi \: (\pred \: \tc{\alpha} \: y) \: (f' \: y) , a , f] } ) \\
    &= \rU_{f^{\rV} [\idx \: \tc{\alpha} , \lambda y . \Phi \: (\pred \: \tc{\alpha} \: y) \: (f' \: y) , a , f] } .
  \end{align*}
  In this proof, we abbreviate $f^{\rV} [\idx \: \tc{\alpha} , \lambda y . \Phi \: (\pred \: \tc{\alpha} \: y) \: (f' \: y) , a , f]$ as $f^{\rV}$. Recall that $f^{\rV}$ is of the form
  \[
  \lambda w . (\wht{\rN_1}_{\rV} \: \wht{+}_{\rV} \: a \: \wht{+}_{\rV} \: (\idx \: \tc{\alpha}) \: \wht{+}_{\rV} \: \wht{\Sigma}_{\rV} (\idx \: \tc{\alpha} , \lambda w' . \rp_1 ((\lambda y . \Phi \: (\pred \: \tc{\alpha} \: y) \: (f' \: y)) \: w' \: w)) , \: h^{\mathbb{M}} \: w)
  \]
  by Definition \ref{def:uop}.
  
  On the other hand, the following holds for $\res{1}{f^{\rV}} (b , g , \ri \rj \: u) : \iteM{\alpha}{t}{a}{f}$.
  \begin{align*}
    \iteTh{\alpha}{t}{a}{f} (\res{1}{f^{\rV}} (b , g , \ri \rj \: u)) &= \rp_2 (f^{\rV} (b^{\ast} , g^{\ast})) \: (\ri \rj \: u) \\
    &= h^{\rV} (b^{\ast} , g^{\ast}) \: (\ri \rj \: u) = \wht{\mathcal{U}}
  \end{align*}
  Therefore, $\res{1}{f^{\rV}} (b , g , \ri \rj \: u)$ is a code of $\wht{\mathcal{U}}$ in $\iteM{\alpha}{t}{a}{f}$. Put $\tld{\mathcal{U}} := \res{1}{f^{\rV}} (b , g , \ri \rj \: u)$. Then, by defining
  \[
  \tld{\rT} := \lambda z . \res{1}{f^{\rV}}(b , g , \rj (u , z) ) : \iteT{\alpha}{t}{a}{f} (\tld{\mathcal{U}}) \to \iteM{\alpha}{t}{a}{f} ,
  \]
  we have $\tld{\cV} := \wht{\rW}_{f^{\rV}}(\tld{\mathcal{U}} , \tld{\rT}) : \iteM{\alpha}{t}{a}{f}$.
  
  Next, we show that $\xi \in \itebV{\pred \: \tc{\alpha} \: u}{f' \: u}{b^{\ast}}{g^{\ast}}$ holds. Let $\wht{\bV} := \itebV{\pred \: \tc{\alpha} \: u}{f' \: u}{b^{\ast}}{g^{\ast}}$. It suffices to verify that for any $\delta : \bbV$, if $\delta \in \beta$ holds then $\delta \in \wht{\bV}$ holds, because we have $\xi \in \beta$.

  We prove the claim by transfinite induction on $\delta$. Assume that $\delta \in \beta$ holds, so we have $\delta \doteq \iteh{\alpha}{t}{a}{f}(\pred \: x \: z)$ for some $z : \iteT{\alpha}{t}{a}{f} (b)$. Let $C' := g \: z$. For any $v : \iteT{\alpha}{t}{a}{f} (C')$, we have
  \begin{align*}
    \iteh{\alpha}{t}{a}{f} (\pred \: (\pred \: x \: z) \: v) &\in \wcons (\iteTh{\alpha}{t}{a}{f}(\idx (\pred \: x \: z)) , \lambda y . \iteh{\alpha}{t}{a}{f} (\pred \: (\pred \: x \: z) \: y) ) \\
    &= \iteh{\alpha}{t}{a}{f} \wcons (\idx \: (\pred \: x \: z) , \pred \: (\pred \: x \: z)) \\
    &\doteq \iteh{\alpha}{t}{a}{f} (\pred \: x \: z) \\
    &\doteq \delta .
  \end{align*}
  Since $\beta$ is transitive, we have $\iteh{\alpha}{t}{a}{f} (\pred \: (\pred \: x \: z) \: v) \in \beta$, so $\iteh{\alpha}{t}{a}{f} (\pred \: (\pred \: x \: z) \: v) \in \wht{\bV}$ by IH of induction on $\delta$. Thus,
  \[
  \Pi_{(v : \iteT{\alpha}{t}{a}{f} (C'))} \Sigma_{(y : \iteV{\pred \: \tc{\alpha} \: u}{f' \: u}{b^{\ast}}{g^{\ast}})} \iteh{\alpha}{t}{a}{f} (\pred \: (\pred \: x \: z) \: v) \doteq \iteh{\pred \: \tc{\alpha} \: u}{f' \: u}{b^{\ast}}{g^{\ast}} y
  \]
  holds, hence we have by the type-theoretic axiom of choice provable in $\MLTT$
  \begin{align}
    \Sigma_{(g' : \iteT{\alpha}{t}{a}{f} (C') \to \iteV{\pred \: \tc{\alpha} \: u}{f' \: u}{b^{\ast}}{g^{\ast}})} \Pi_{(v : \iteT{\alpha}{t}{a}{f} (C'))} \iteh{\alpha}{t}{a}{f} (\pred \: (\pred \: x \: z) \: v) \doteq \iteh{\pred \: \tc{\alpha} \: u}{f' \: u}{b^{\ast}}{g^{\ast}} (g' \: v) . \label{eq:byAC}
  \end{align}

  On the other hand, by Lemma \ref{lem:code}, there is a term $C$ of type $\iteM{\pred \: \tc{\alpha} \: u}{f' \: u}{b^{\ast}}{g^{\ast}}$ such that $\iteTh{\pred \: \tc{\alpha} \: u}{f' \: u}{b^{\ast}}{g^{\ast}} (C) =_{\rV} g^{\ast} \: z$ holds. Since $g^{\ast} \: z = \iteTh{\alpha}{t}{a}{f} (g \: z) = \iteTh{\alpha}{t}{a}{f} (C')$ holds, we have a term $p$ of type $\iteTh{\pred \: \tc{\alpha} \: u}{f' \: u}{b^{\ast}}{g^{\ast}} (C) =_{\rV} \iteTh{\alpha}{t}{a}{f} (g \: z)$. Therefore, we also have a term
  \[
  \wcons (C , \lambda y . g' \: (\mathsf{transport}^{\rTV (v)} \: p \: y)) : \iteV{\pred \: \tc{\alpha} \: u}{f' \: u}{b^{\ast}}{g^{\ast}},
  \]
  because for any $y : \iteTh{\pred \: \tc{\alpha} \: u}{f' \: u}{b^{\ast}}{g^{\ast}} (C)$, we have $g' \: (\mathsf{transport}^{\rTV (v)} \: p \: y) : \iteV{\pred \: \tc{\alpha} \: u}{f' \: u}{b^{\ast}}{g^{\ast}}$. So
  \[
  \iteh{\pred \: \tc{\alpha} \: u}{f' \: u}{b^{\ast}}{g^{\ast}} \wcons (C , \lambda y . g' \: (\mathsf{transport}^{\rTV (v)} \: p \: y)) \in \wht{\bV} = \itebV{\pred \: \tc{\alpha} \: u}{f' \: u}{b^{\ast}}{g^{\ast}}
  \]
  holds.

  Then, by \eqref{eq:byAC} and Lemma \ref{lem:transp},
  \begin{align*}
    &\; \iteh{\pred \: \tc{\alpha} \: u}{f' \: u}{b^{\ast}}{g^{\ast}} \wcons (C , \lambda y . g' \: (\mathsf{transport}^{\rTV (v)} \: p \: y)) \\
    &\doteq \wcons (\iteTh{\pred \: \tc{\alpha} \: u}{f' \: u}{b^{\ast}}{g^{\ast}} (C) , \lambda y . \iteh{\pred \: \tc{\alpha} \: u}{f' \: u}{b^{\ast}}{g^{\ast}} (g' \: (\mathsf{transport}^{\rTV (v)} \: p \: y))) \\
    &\doteq \wcons (\iteTh{\pred \: \tc{\alpha} \: u}{f' \: u}{b^{\ast}}{g^{\ast}} (C) , \lambda y . \iteh{\alpha}{t}{a}{f} (\pred \: (\pred \: x \: z) \: (\mathsf{transport}^{\rTV (v)} \: p \: y))) \\
    &\doteq \wcons (\iteTh{\alpha}{t}{a}{f} (C ') , \lambda y . \iteh{\alpha}{t}{a}{f} (\pred \: (\pred \: x \: z) \: y) ) \\
    &\doteq \delta
  \end{align*}
  holds, so we have $\delta \in \wht{\bV}$.

  Finally, we show that $\wht{\bV} \in \itebV{\alpha}{t}{a}{f}$ holds. Recall that $\tld{\cV}$ is of type $\iteM{\alpha}{t}{a}{f}$, and put
  \[
  \cV := \iteT{\alpha}{t}{a}{f} (\tld{\cV}) = \rW_{(v : \wht{\mathcal{U}})} \iteT{\pred \: \tc{\alpha} \: u}{f' \: u}{b^{\ast}}{g^{\ast}} (v) .
  \]
  By recursion on $\cV$, we define $l : \cV \to \iteV{\alpha}{t}{a}{f}$ as
  \[
  l \: \wcons (d , q) := \wcons (\tld{\rT} \: d , \lambda y . l \: (q \: y)) ,
  \]
  and put $\bV ' := \wcons (\tld{\cV} , l) : \iteV{\alpha}{t}{a}{f}$.

  We then define $\beta ' := \iteh{\alpha}{t}{a}{f} \bV ' : \bbV$, so $\beta ' \in \itebV{\alpha}{t}{a}{f}$ holds. It suffices to show that we have $\beta ' \doteq \wht{\bV}$.

  We first show that
  \begin{align}
    \Pi_{(v : \cV)} \iteh{\alpha}{t}{a}{f} (l \: v) \doteq \iteh{\pred \: \tc{\alpha} \: u}{f' \: u}{b^{\ast}}{g^{\ast}} v \label{eq:ell}
  \end{align}
  holds by induction on $v = \wcons (c , q)$. It follows from the induction hypothesis that we have
  \[
  \Pi_{(y : \iteT{\pred \: \tc{\alpha} \: u}{f' \: u}{b^{\ast}}{g^{\ast}} c)} \iteh{\alpha}{t}{a}{f} (l \: (q \: y)) \doteq \iteh{\pred \: \tc{\alpha} \: u}{f' \: u}{b^{\ast}}{g^{\ast}} (q \: y) .
  \]
  So,
  \begin{align*}
    \iteh{\alpha}{t}{a}{f}(l \: v) &= \iteh{\alpha}{t}{a}{f} \wcons (\tld{\rT} \: c , \lambda y . l \: (q \: y)) \\
    &= \wcons (\iteTh{\alpha}{t}{a}{f} (\tld{\rT} \: c) , \lambda y . \iteh{\alpha}{t}{a}{f} (l \: (q \: y))) \\
    &\doteq \wcons (\iteTh{\alpha}{t}{a}{f} (\tld{\rT} \: c) , \lambda y . \iteh{\pred \: \tc{\alpha} \: u}{f' \: u}{b^{\ast}}{g^{\ast}} (q \: y)) \\
    &= \iteh{\pred \: \tc{\alpha} \: u}{f' \: u}{b^{\ast}}{g^{\ast}} v
  \end{align*}
  holds. The last equation here holds because
  \begin{align*}
    \iteTh{\pred \: \tc{\alpha} \: u}{f' \: u}{b^{\ast}}{g^{\ast}} (c) &= \rp_2 (\iteHh{\pred \: \tc{\alpha} \: u}{f' \: u}{b^{\ast}}{g^{\ast}}) \: c \\
    &= h^{\rV} (b^{\ast} , g^{\ast}) \: (\rj (u , c)) \text{ with } u : \ov{\tc{\alpha}} \\
    &= \rp_2 (f^{\rV} (b^{\ast} , g^{\ast})) \: \rj (u , c) \\
    &= \iteTh{\alpha}{t}{a}{f} (\res{1}{f^{\rV}} (b , g , \rj (u , c))) \\
    &= \iteTh{\alpha}{t}{a}{f} (\tld{\rT} \: c) .
  \end{align*}
  Therefore, we have \eqref{eq:ell}.

  Since $\beta ' = \wcons (\iteTh{\alpha}{t}{a}{f} (\tld{\cV}) , \lambda y . \iteh{\alpha}{t}{a}{f} (l \: y))$ holds, if we have $\eta \in \beta '$ then $\eta \doteq \iteh{\alpha}{t}{a}{f} (l \: v)$ holds for some $v : \cV$. We thus have $\eta \doteq \iteh{\alpha}{t}{a}{f} (l \: v) \doteq \iteh{\pred \: \tc{\alpha} \: u}{f' \: u}{b^{\ast}}{g^{\ast}} v \in \wht{\bV}$ by \eqref{eq:ell}.

  Conversely, if $\eta \in \wht{\bV}$ holds then $\eta \doteq \iteh{\pred \: \tc{\alpha} \: u}{f' \: u}{b^{\ast}}{g^{\ast}} v$ holds for some $v : \cV$. So, we have
  \begin{align*}
    \eta &\doteq \iteh{\pred \: \tc{\alpha} \: u}{f' \: u}{b^{\ast}}{g^{\ast}} v \\
    &\doteq \iteh{\alpha}{t}{a}{f}(l \: v) \text{ by } \eqref{eq:ell} \\
    &\in \beta ' .
  \end{align*}
  It follows that $\beta ' \doteq \wht{\bV}$.\qedhere
\end{proof}

We now prove the lemma saying that for any $\alpha : \bbV$ such that $\Acc \: \alpha$ holds with $a : \rV$ and $f : \rTV (a) \to \rV$, $\itebV{\alpha}{t}{a}{f}$ is $\alpha$-set-inaccessible.
\begin{lem}[\href{https://github.com/takahashi-yt/czf-in-mahlo/blob/67c9dbf50bee67b73c1ea009929f4838458b8edc/src/PiNumbersAxiom.agda\#L162}{\texttt{Link}}]\label{lem:lastone}
  For all $\alpha : \bbV$, $t : \Acc \: \alpha$, $a : \rV$ and $f : \rTV (a) \to \rV$, $\itebV{\alpha}{t}{a}{f}$ is $\alpha$-set-inaccessible. 
\end{lem}
\begin{proof}
  Induction on $\alpha$ followed by subinduction on $t$. Let $\beta \in \alpha$ be the case. We then have $\beta \in \tc{\alpha}$, so $\beta \doteq \pred \: \tc{\alpha} \: u$ for some $u : \ov{\tc{\alpha}}$.

  If $\xi \in \itebV{\alpha}{\prog \: f'}{a}{f}$ holds then Lemma \ref{lem:main} implies that there are terms $b : \iteM{\alpha}{\prog \: f'}{a}{f}$ and $g : \iteT{\alpha}{\prog \: f'}{a}{f} (b) \to \iteM{\alpha}{\prog \: f'}{a}{f}$ with
  \[
  \xi \in \itebV{\pred \: \tc{\alpha} \: u}{f' \: u}{b^{\ast}}{g^{\ast}} \in \itebV{\alpha}{\prog \: f'}{a}{f} ,
  \]
  where $b^{\ast} := \iteTh{\alpha}{\prog \: f'}{a}{f} (b)$ and $g^{\ast} := \lambda x . \iteTh{\alpha}{\prog \: f'}{a}{f}(g \: x)$.

  By IH of main induction, $\itebV{\pred \: \tc{\alpha} \: u}{f' \: u}{b^{\ast}}{g^{\ast}}$ is $(\pred \: \tc{\alpha} \: u)$-set-inaccessible. Since $\alpha$-set-inaccessibility is the translation of a corresponding set-theoretic formula for any $\alpha : \bbV$, Lemma \ref{lem:czfext} implies that if $\delta \doteq \eta$ holds and $\gamma$ is $\delta$-set-inaccessible then $\gamma$ is also $\eta$-set-inaccessible. Therefore, $\itebV{\pred \: \tc{\alpha} \: u}{f' \: u}{b^{\ast}}{g^{\ast}}$ is also $\beta$-set-inaccessible, and it follows that $\beta$-set-inaccessibles are unbounded in $\itebV{\alpha}{\prog \: f'}{a}{f}$. By Lemmas \ref{lem:setinacc} and \ref{lem:lemsone}.(5), $\itebV{\alpha}{\prog \: f'}{a}{f}$ is $\alpha$-set-inaccessible.
\end{proof}

Below we use the following notations:
  \begin{align*}
    \wht{\mathcal{U}}^{\alpha , t}_{\beta} &:= \rp_1 (\iteHh{\alpha}{t}{\idx \: \beta}{\lambda x . \idx \: (\pred \: \beta \: x)}) & \mathcal{U}^{\alpha , t}_{\beta} &:= \rTV \wht{\mathcal{U}}^{\alpha , t}_{\beta} \\
    \wht{\rT}^{\alpha , t}_{\beta} &:= \rp_2 (\iteHh{\alpha}{t}{\idx \: \beta}{\lambda x . \idx \: (\pred \: \beta \: x)}) & \rT^{\alpha , t}_{\beta} &:= \lambda x . \rTV (\wht{\rT}^{\alpha , t}_{\beta} x) \\
    \wht{\mathcal{V}}^{\alpha , t}_{\beta} &:= \wht{\rW}_{\rV} ( \wht{\mathcal{U}}^{\alpha , t}_{\beta} , \wht{\rT}^{\alpha , t}_{\beta} ) & \mathcal{V}^{\alpha , t}_{\beta} &:= \rTV \wht{\mathcal{V}}^{\alpha , t}_{\beta}
  \end{align*}
We also define $\mathbf{h}^{\alpha , t}_{\beta}$ and $\bV^{\alpha , t}_{\beta}$ as
  \[
  \mathbf{h}^{\alpha , t}_{\beta} \wcons (b , g) := \wcons (\wht{\rT}^{\alpha , t}_{\beta} (b) , \lambda x . \mathbf{h}^{\alpha , t}_{\beta} (g \: x) ), \qquad \bV^{\alpha , t}_{\beta} := \wcons (\wht{\mathcal{V}}^{\alpha , t}_{\beta} , \mathbf{h}^{\alpha , t}_{\beta} ),
  \]
  respectively.
\begin{lem}[\href{https://github.com/takahashi-yt/czf-in-mahlo/blob/67c9dbf50bee67b73c1ea009929f4838458b8edc/src/PiNumbersAxiom.agda\#L218}{\texttt{Link}}]\label{lem:lasttwo}
  For any $\beta : \bbV$, if $\beta$ is transitive, then for any $\alpha : \bbV$ and any $t : \Acc \: \alpha$, we have $\Pi_{(\gamma : \bbV)} (\gamma \in \beta \to \gamma \in \bV^{\alpha , t}_{\beta})$.
\end{lem}
\begin{proof}
  We prove $\Pi_{(\gamma : \bbV)} (\gamma \in \beta \to \gamma \in \bV^{\alpha , t}_{\beta})$ by induction on $\gamma$. Suppose that $\gamma \in \beta$ holds. Then, $\gamma \doteq \pred \: \beta \: x$ holds for some $x : \ov{\beta}$. Put $\xi := \pred \: \beta \: x$. If $u : \ov{\xi}$ holds, then we have $\pred \: \xi \: u \in \gamma \in \beta$. Since $\beta$ is transitive, we also have $\pred \: \xi \: u \in \beta$ and so $\pred \: \xi \: u \in \bV^{\alpha , t}_{\beta}$ holds by IH. We thus have
  \[
  \Pi_{(u : \ov{\xi})} \Sigma_{(z : \mathcal{V}^{\alpha , t}_{\beta})} \pred \: \xi \: u \doteq \mathbf{h}^{\alpha , t}_{\beta}z .
  \]
  By the type-theoretic axiom of choice in $\MLTT$,
  \[
  \Sigma_{(f : \ov{\xi} \to \mathcal{V}^{\alpha , t}_{\beta})} \Pi_{(u : \ov{\xi})}  \pred \: \xi \: u \doteq \mathbf{h}^{\alpha , t}_{\beta} (f \: u)
  \]
  holds.

  On the other hand, by Lemma \ref{lem:code}, there is a term $c : \mathcal{U}^{\alpha , t}_{\beta}$ such that
  \[
  p : \wht{\rT}^{\alpha , t}_{\beta} (c) =_{\rV} (\lambda z . \idx \: (\pred \: \beta \: z)) x
  \]
  holds. Then, it follows from Lemma \ref{lem:transp} that
  \[
  \wcons (\wht{\rT}^{\alpha , t}_{\beta} (c) , \lambda x . \mathbf{h}^{\alpha , t}_{\beta} (f \: (\mathsf{transport}^{\rTV (v)} \: p \: x))) \doteq \wcons (\idx \: \xi , \lambda u . \mathbf{h}^{\alpha , t}_{\beta} (f \: u))
  \]
  holds. Note that we have
  \[
  \wcons (c , \lambda x . f \: (\mathsf{transport}^{\rTV (v)} \: p \: x)) : \mathcal{V}^{\alpha , t}_{\beta} , \quad \mathbf{h}^{\alpha , t}_{\beta} \wcons (c , \lambda x . f \: (\mathsf{transport}^{\rTV (v)} \: p \: x)) \in \bV^{\alpha , t}_{\beta} .
  \]
  Moreover,
  \begin{align*}
    \mathbf{h}^{\alpha , t}_{\beta} \wcons (c , \lambda x . f \: (\mathsf{transport}^{\rTV (v)} \: p \: x)) &= \wcons (\wht{\rT}^{\alpha , t}_{\beta}c , \lambda x . \mathbf{h}^{\alpha , t}_{\beta} (f \: (\mathsf{transport}^{\rTV (v)} \: p \: x))) \\
    &\doteq \wcons (\idx \: \xi , \lambda u . \mathbf{h}^{\alpha , t}_{\beta} (f \: u)) \\
    &\doteq \wcons (\idx \: \xi , \pred \: \xi) \\
    &\doteq \xi 
  \end{align*}
  holds, hence we have $\gamma \doteq \xi \in \bV^{\alpha , t}_{\beta}$.
\end{proof}

\begin{lem}[\href{https://github.com/takahashi-yt/czf-in-mahlo/blob/67c9dbf50bee67b73c1ea009929f4838458b8edc/src/TypeHierarchy.agda}{\texttt{Link}}]\label{lem:lastthree}
  The following holds:
  \begin{enumerate}
  \item We have a term $\mathrm{inv} : \Pi_{(\alpha : \bbV)} (\Acc \: \alpha \to \Pi_{(x : \ov{\tc{\alpha}})} \Acc \: (\pred \: \tc{\alpha} \: x))$.

  \item We have a function $\mathrm{acc} :  \Pi_{(\alpha : \bbV)} \Acc \: \alpha$.
    
  \end{enumerate}
\end{lem}
\begin{proof}
  (1) By induction on $t : \Acc \: \alpha$. Let $t$ be $\prog \: f$, then we have $f \: x : \Acc \: (\pred \: \tc{\alpha} \: x)$.

  (2) By induction on $\alpha : \bbV$. Let $\alpha$ be $\wcons (a , f)$, and recall that $\ov{\tc{(\wcons (a , f))}}$ is definitionally equal to $\rTV (a) + \Sigma_{(z : \rTV (a))} \ov{\tc{(f \; z)}}$. Define $g : \Pi_{(x : \ov{\tc{(\wcons (a , f))}})} \Acc \: (\pred \: \tc{(\wcons (a , f))} \: x)$ by
  \begin{align*}
    g \: (\ri \: z) &= \mathrm{acc} \: (f \: z) , \\
    g \: (\rj \: w) &= \mathrm{inv} \: (f \: (\rp_1 \: w)) \: (\mathrm{acc} \: (f \: (\rp_1 \: w))) \: (\rp_2 \: w) .
  \end{align*}
  The function $g$ is well-defined because we have
  \[
  \pred \: \tc{(\wcons (a , f))} \: (\ri \: z) = f \: z , \qquad \pred \: \tc{(\wcons (a , f))} \: (\rj \: w) = \pred \: \tc{(f \: (\rp_1 \: w))} \: (\rp_2 \: w) .
  \]
  We then put $\mathrm{acc} \: (\wcons (a , f)) := \prog \: g$.
\end{proof}

We have an interpretation of \textbf{Pi-Numbers} in $\MLMacc$ by the theorem below:
\begin{thm}[\href{https://github.com/takahashi-yt/czf-in-mahlo/blob/67c9dbf50bee67b73c1ea009929f4838458b8edc/src/PiNumbersAxiom.agda\#L267}{\texttt{Link}}]
  For any $\gamma , \alpha : \bbV$, there is a term $\beta : \bbV$ such that $\gamma \in \beta$ holds and $\beta$ is $\alpha$-set-inaccessible.
\end{thm}
\begin{proof}
  Let $\gamma , \alpha : \bbV$ be given. It is straightforward to define a term $\gamma '$ of $\bbV$ with $\gamma \in \gamma '$, since all axioms of $\CZF$ are interpretable in $\MLTT$. By Lemmas \ref{lem:lasttwo} and \ref{lem:lastthree}.(2), we have
  \[
  \Pi_{(\delta : \bbV)} (\delta \in \tc{\gamma '} \to \delta \in \bV^{\alpha , \mathrm{acc} \: \alpha}_{\tc{\gamma '}}) .
  \]
  Therefore, $\gamma \in \bV^{\alpha , \mathrm{acc} \: \alpha}_{\tc{\gamma '}}$ holds, and $\bV^{\alpha , \mathrm{acc} \: \alpha}_{\tc{\gamma '}}$ is $\alpha$-set-inaccessible by Lemma \ref{lem:lastone}.
\end{proof}

The statements whose proofs are not formalised yet in \cite{takahashi2025} are as follows:
\begin{enumerate}
\item $\itebV{\alpha}{t}{a}{f}$ satisfies \textbf{Regular Extension Axiom} (\cf~\cite{RathjenGrifforPalmgren1998}). \label{todo:one}

\item For any $\alpha , \beta : \bbV$, if $\alpha \doteq \beta$ holds and $\alpha$ is set-inaccessible then $\beta$ is set-inaccessible. \label{todo:two}

\item For any $\alpha , \beta , \gamma : \bbV$, if $\alpha \doteq \beta$ holds and $\gamma$ is $\alpha$-set-inaccessible then $\gamma$ is $\beta$-set-inaccessible. \label{todo:three}
  
\end{enumerate}
The statement~\ref{todo:one} is postulated in \href{https://github.com/takahashi-yt/czf-in-mahlo/blob/67c9dbf50bee67b73c1ea009929f4838458b8edc/src/RegularExtensionAxiom.agda}{\texttt{RegularExtensionAxiom.agda}} of \cite{takahashi2025}, while the other statements are postulated in \href{https://github.com/takahashi-yt/czf-in-mahlo/blob/67c9dbf50bee67b73c1ea009929f4838458b8edc/src/Inaccessibles/Basics.agda}{\texttt{Inaccessibles/Basics.agda}}. The statement~\ref{todo:one} is used to prove Lemma \ref{lem:lemsone}.(5), and we referred to \cite[Corollary 4.8]{RathjenGrifforPalmgren1998} for its proof. Postulating the statements~\ref{todo:two} and \ref{todo:three} is harmless, since we have Lemma \ref{lem:czfext} for the $\MLTT$-translation of the set-theoretic terms ``set-inaccessible'' and ``$\alpha$-set-inaccessible''.


\section{Concluding Remark}
In this paper, we showed that the system $\CZF_{\pi}$ of Aczel's constructive set theory with inaccessible sets is interpretable in a variant of $\MLTT$ equipped with a Mahlo universe and accessibility predicate $\Acc$. Our type-theoretic counterparts of inaccessible sets were constructed by the transfinite iteration of a higher-order universe operator whose definition used the reflection property of the Mahlo universe. The accessibility predicate $\Acc$ provided a recursion principle to iterate this operator. Moreover, we formalised the main part of this interpretation in Agda. For this purpose, the support of indexed induction-recursion in Agda is crucial: the Mahlo universe was simulated by the external Mahlo universe defined with induction-recursion, and the predicate $\Acc$ was defined by indexed inductive definition.

A future research direction is to investigate the extent to which the transfinite hierarchy of inaccessible sets can be extended in $\MLTT$. In set-theoretic terms, the diagonalisation of this hierarchy gives the class $\{ a \mid \text{$a$ is $a$-set-inaccessible}\}$ (\cf~\cite{RichterAczel1974}). Members of this class can be considered as ``fixed points'' of an increasing sequence from inaccessible sets. It would be interesting to examine whether $\MLM$ can construct an element of $\bbV$ corresponding to a member of this class and, if possible, whether such elements of $\bbV$ are unbounded in $\bbV$.

Another direction is to revisit our construction of inaccessible sets in $\MLM$ from the viewpoint of homotopy type theory (HoTT). In \cite{hottbook2013}, Aczel's interpretation of $\CZF$ in $\MLTT$ was refined: the type of Aczel's iterative sets was defined as a higher inductive type $\bbV_h$. In addition to its ordinary introduction rule (\textit{i.e.}, its point constructor rule), which is the same as that of $\bbV$, the higher inductive type $\bbV_h$ also has path constructor rules. Although it is unknown whether the axioms \textbf{Strong Collection} and \textbf{Subset Collection} of $\CZF$ hold in this refined interpretation, several interesting results on this interpretation have already been obtained (see, \eg, \cite{JKFX2023}). Later, a general framework for interpreting constructive set theory into HoTT was proposed in \cite{gallozzi2021}. Importantly, all axioms of $\CZF$ hold in some interpretations of this framework; therefore, it would provide a HoTT framework to revisit our interpretation of $\CZF_{\pi}$ in $\MLMacc$.

\section*{Acknowledgement}
  \noindent The author is grateful to the anonymous reviewers for their comments and suggestions, which improved the contents and presentation of the present paper. The work in this paper is supported by JSPS KAKENHI Grant Number JP21K12822.



\bibliographystyle{alphaurl}
\bibliography{mahlo_czf_final}

\appendix

\section{Inference Rules of $\MLM$}\label{sec:mlmrules}
In this appendix, we define the inference rules of $\MLM$ in detail to make the present paper self-contained as far as possible. We provided the Agda code corresponding to the rules for the Mahlo universe $\rV$ in \href{https://github.com/takahashi-yt/czf-in-mahlo/blob/67c9dbf50bee67b73c1ea009929f4838458b8edc/src/ExternalMahlo.agda}{\texttt{ExternalMahlo.agda}} of \cite{takahashi2025}.
\begin{defi}[Inference Rules and Judgements in $\MLM$]
  Judgements in $\MLM$ are those which are derivable by the following inference rules of $\MLM$:
  \begin{description}
  \item[Rules for Contexts] let $x$ be fresh for $\Gamma$, and $y$ be fresh for $\Delta$.
    \[
    \infer{\vdash \emptyset \; \ctxt}{}
    \qquad
    \infer{\vdash \Gamma , x : A \; \ctxt}{\Gamma \vdash A \; \type}
    \]

  \item[Rule for Assumptions]
    \[
    \infer{ \Gamma , x : A , \Delta \vdash x : A }{ \vdash \Gamma , x : A , \Delta \; \ctxt }
    \]

  \item[Rules for Substitutions]
    \[
    \infer{\Gamma , \Delta [a / x] \vdash \theta [a / x]}{
      \Gamma \vdash a : A
      &
      \Gamma , x : A , \Delta \vdash \theta
    }
    \qquad
    \infer{\Gamma , \Delta [a / x] \vdash B [a / x] = B[b / x]}{
      \Gamma \vdash a = b : A
      &
      \Gamma , x : A , \Delta \vdash B \; \type
    }
    \]
    \[
    \infer{\Gamma , \Delta [a / x] \vdash c [a / x] = c [b / x] : B [a / x]}{
      \Gamma \vdash a = b : A
      &
      \Gamma , x : A , \Delta \vdash c : B
    }
    \]

    Hereafter, we list the premises of an inference rule in the vertical way.

  \item[Rules for Definitional Equality]
    \[
    \infer{ \Gamma \vdash a = a : A }{ \Gamma \vdash a : A }
    \qquad
    \infer{ \Gamma \vdash b = a : A }{ \Gamma \vdash a = b : A }
    \qquad
    \infer{ \Gamma \vdash a = c : A }{
      \deduce{\Gamma \vdash b = c : A}{\Gamma \vdash a = b : A}
    }
    \]
    \[
    \infer{ \Gamma \vdash A = A}{ \Gamma \vdash A \; \type }
    \qquad
    \infer{ \Gamma \vdash B = A}{ \Gamma \vdash A = B}
    \qquad
    \infer{ \Gamma \vdash A = C}{
      \deduce{\Gamma \vdash B = C}{\Gamma \vdash A = B}
    }
    \qquad
    \infer{\Gamma \vdash a : B}{
      \deduce{\Gamma \vdash A = B}{\Gamma \vdash a : A}
    }
    \]
    
  \item[Rules for $\Pi$-Types] we have the formation rule $\Pi \fm$, the introduction rule $\Pi \intro$, the elimination rule $\Pi \el$ and the $\beta$/$\eta$ equality rules $\Pi \eq \beta$/$\Pi \eq \eta$ for $\Pi$-types.
    \[
    \infer[\Pi \fm]{\Gamma \vdash \Pi_{(x : A)} B \; \type}{
      \deduce{\Gamma , x : A \vdash B \; \type}{\Gamma \vdash A \; \type}
    }
    \qquad
    \infer[\Pi \intro]{\Gamma \vdash \lambda x . b : \Pi_{(x : A)} B}{\Gamma , x : A \vdash b : B}
    \]
    \[
    \infer[\Pi \el]{\Gamma \vdash b\: a : B [a / x]}{
      \deduce{\Gamma \vdash a : A}{
        \deduce{\Gamma \vdash b : \Pi_{(x : A)} B}{\Gamma , x : A \vdash B \; \type}
      }
    }
    \qquad
    \infer[\Pi \eq \beta]{\Gamma \vdash (\lambda x . b)a = b[a/x] : B [a/x]}{
      \deduce{\Gamma \vdash a : A}{\Gamma , x : A \vdash b : B}
    }
    \]
    \[
    \infer[\Pi \eq \eta \quad \text{ with $y$ not free in $b$}]{\Gamma \vdash \lambda y . b \: y = b : \Pi_{(x : A)} B}{
      \deduce{\Gamma \vdash b : \Pi_{(x : A)} B}{
        \Gamma , x : A \vdash B \; \type
      }
    }
    \]
    Moreover, we have a congruence rule for each of $\Pi \fm , \Pi \intro , \Pi \el$ as follows.
    \[
    \infer[\Pi \fm_{c}]{\Gamma \vdash \Pi_{(x : A_1)} B_1 = \Pi_{(x : A_2)} B_2}{
      \deduce{\Gamma , x : A_1 \vdash B_1 = B_2}{\Gamma \vdash A_1 = A_2}
    }
    \qquad
    \infer[\Pi \intro_{c}]{\Gamma \vdash \lambda x . b_1 = \lambda x . b_2 : \Pi_{(x : A)} B}{\Gamma , x : A \vdash b_1 = b_2 : B}
    \]
    \[
    \infer[\Pi \el_{c}]{\Gamma \vdash b_1  a_1 = b_2 a_2 : B [a_1 / x]}{
      \deduce{\Gamma \vdash a_1 = a_2 : A}{
        \deduce{\Gamma \vdash b_1 = b_2 : \Pi_{(x : A)} B}{\Gamma , x : A \vdash B \; \type}
      }
    }
    \]

    As in the case of $\Pi$-types, for any $\ast \in \{ \Sigma , + , \rN_n , \rN , \rW , \Id \}$, we use the labels $\ast \fm , \ast \intro , \ast \el$ and $\ast \eq$ possibly with suffixes to denote $\ast$-formation rules, $\ast$-introduction rules, $\ast$-elimination rules and $\ast$-equality rules, respectively. Hereafter, we omit to write down congruence rules: we assume that congruence rules are formulated for the $\ast$-formation, $\ast$-introduction and $\ast$-elimination rules.

  \item[Rules for $\Sigma$-Types] we have the following rules.
    \[
    \infer[\Sigma \fm]{\Gamma \vdash \Sigma_{(x : A)} B \; \type}{
      \deduce{\Gamma , x : A \vdash B \; \type}{\Gamma \vdash A \; \type}
    }
    \quad
    \infer[\Sigma \intro]{\Gamma \vdash ( a , b ) : \Sigma_{(x : A)} B}{
      \deduce{\Gamma \vdash b : B[a / x]}{
        \deduce{\Gamma \vdash a : A}{\Gamma , x : A \vdash B \; \type}
      }
    }
    \quad
    \infer[\Sigma \el]{\Gamma \vdash \rE_{\Sigma} ( b , c ) : C [b / z]}{
      \deduce{\Gamma \vdash b : \Sigma_{(x : A)} B}{
        \deduce{\Gamma \vdash c : \Pi_{(x : A)}\Pi_{(y : B)} C [ ( x , y ) /z] }{\Gamma , z : \Sigma_{(x : A)} B \vdash C \; \type}
      }
    }
    \]
    \[
    \infer[\Sigma \eq]{\Gamma \vdash \rE_{\Sigma} (( a , b ) , c ) = c \:a \: b : C [( a , b ) / z ]}{
      \deduce{\Gamma \vdash c : \Pi_{(x : A)}\Pi_{(y : B)} C [( x , y ) /z] }{
        \deduce{\Gamma , z : \Sigma_{(x : A)} B \vdash C \; \type}{
          \deduce{\Gamma \vdash b : B [ a / x ]}{
            \deduce{\Gamma \vdash a : A}{\Gamma , x : A \vdash B \; \type}
          }
        }
      }
    }
    \]
    It is straightforward to define the left projection $\rp_1$ and the right projection $\rp_2$ satisfying the following:
    \[
    \infer{\Gamma \vdash \rp_1 c : A}{
      \deduce{\Gamma \vdash c : \Sigma_{(x : A)} B}{\Gamma , x : A \vdash B \; \type}
    }
    \qquad
    \infer{\Gamma \vdash \rp_1 ( a , b ) = a : A}{
      \deduce{\Gamma \vdash b : B[a / x]}{
        \deduce{\Gamma \vdash a : A}{\Gamma , x : A \vdash B \; \type}
      }
    }
    \]
    \[
    \infer{\Gamma \vdash \rp_2 c : B[\rp_1 (c) / x]}{
      \deduce{\Gamma \vdash c : \Sigma_{(x : A)} B}{\Gamma , x : A \vdash B \; \type}
    }
    \qquad
    \infer{\Gamma \vdash \rp_2 ( a , b ) = b : B[\rp_1 (a , b) / x]}{
      \deduce{\Gamma \vdash b : B[a / x]}{
        \deduce{\Gamma \vdash a : A}{\Gamma , x : A \vdash B \; \type}
      }
    }
    \]

  \item[Rules for $+$-Types] a sum type $A + B$ corresponds to a disjoint sum of $A$ and $B$.
    \[
    \infer[+ \fm]{\Gamma \vdash A + B \; \type}{
      \deduce{\Gamma \vdash B \; \type}{\Gamma \vdash A \; \type}
    }
    \qquad
    \infer[+ \intro_{\ri}]{\Gamma \vdash \ri ( a ) : A + B}{
      \deduce{\Gamma \vdash B \; \type}{\Gamma \vdash a : A}
    }
    \qquad
    \infer[+ \intro_{\rj}]{\Gamma \vdash \rj ( b ) : A + B}{
      \deduce{\Gamma \vdash b : B}{\Gamma \vdash A \; \type}
    }
    \]
    \[
    \infer[+ \el]{\Gamma \vdash \rE_+ ( d , c_1 , c_2 ) : C[d / x]}{
      \deduce{\Gamma \vdash d : A + B}{
        \deduce{\Gamma \vdash c_2 : \Pi_{(z : B)} C[\rj ( z ) / x]}{
          \deduce{\Gamma \vdash c_1 : \Pi_{(y : A)} C[\ri ( y ) / x]}{\Gamma , x : A + B \vdash C \; \type}
        }
      }
    }
    \]
    \[
    \infer[+ \eq_{\ri}]{\Gamma \vdash \rE_+ ( \ri ( a ) , c_1 , c_2 ) = c_1 \: a : C[\ri ( a )/x]}{
      \deduce{\Gamma \vdash a : A}{
        \deduce{\Gamma \vdash c_2 : \Pi_{(z : B)} C[\rj ( z ) / x]}{
          \deduce{\Gamma \vdash c_1 : \Pi_{(y : A)} C[\ri ( y ) / x]}{\Gamma , x : A + B \vdash C \; \type}
        }
      }
    }
    \qquad
    \infer[+ \eq_{\rj}]{\Gamma \vdash \rE_+ ( \rj ( b ) , c_1 , c_2 ) = c_2 \: b : C[\rj ( b )/x]}{
      \deduce{\Gamma \vdash b : B}{
        \deduce{\Gamma \vdash c_2 : \Pi_{(z : B)} C[\rj ( z ) / x]}{
          \deduce{\Gamma \vdash c_1 : \Pi_{(y : A)} C[\ri ( y ) / x]}{\Gamma , x : A + B \vdash C \; \type}
        }
      }
    }
    \]

  \item[Rules for $\rN$-Type] the $\rN$-type is the type of natural numbers. Its elimination rule is the induction principle on $\rN$.
    \[
    \infer[\rN \fm]{\Gamma \vdash \rN \; \type}{\vdash \Gamma \; \ctxt}
    \qquad
    \infer[\rN \intro_{0}]{\Gamma \vdash 0 : \rN}{\vdash \Gamma \; \ctxt}
    \qquad
    \infer[\rN \intro_{\suc}]{\Gamma \vdash \suc ( a ) : \rN}{\Gamma \vdash a : \rN}
    \]
    \[
    \infer[\rN \el]{\Gamma \vdash \rE_{\rN} ( a , b , c  ) : C [ a / x ]}{
      \deduce{\Gamma \vdash c : \Pi_{(x : \rN)}(C \to C [\suc ( x ) / x])}{
        \deduce{\Gamma \vdash b : C[0 / x]}{
          \deduce{\Gamma \vdash a : \rN}{\Gamma , x : \rN \vdash C \; \type}
        }
      }
    }
    \qquad
    \infer[\rN \eq_{0}]{\Gamma \vdash \rE_{\rN} ( 0 , b , c  ) = b : C [ 0 / x ]}{
      \deduce{\Gamma \vdash c : \Pi_{(x : \rN)}(C \to C [\suc ( x ) / x])}{
        \deduce{\Gamma \vdash b : C[0 / x]}{\Gamma , x : \rN \vdash C \; \type}
      }
    }
    \]
    \[
    \infer[\rN \eq_{\suc}]{\Gamma \vdash \rE_{\rN} ( \suc ( a ) , b , c  ) = c \: a \: (\rE_{\rN} ( a , b , c )) : C [ \suc ( a ) / x ]}{
      \deduce{\Gamma  \vdash c : \Pi_{(x : \rN)}(C \to C [\suc ( x ) / x])}{
        \deduce{\Gamma \vdash b : C[0 / x]}{
          \deduce{\Gamma \vdash a : \rN}{\Gamma , x : \rN \vdash C \; \type}
        }
      }
    }
    \]

  \item[Rules for $\rN_n$-Type] the $\rN_n$-type is the type of finite set with $n$ members.
        \[
    \infer[\rN_n \fm]{\Gamma \vdash \rN_n \; \type}{\vdash \Gamma \; \ctxt}
    \qquad
    \infer[\rN_n \intro \quad \text{for each numeral $m < n$}]{\Gamma \vdash m_n : \rN_n}{\vdash \Gamma \; \ctxt}
    \]
    \[
    \infer[\rN_n \el]{\Gamma \vdash \rE_{n} ( a , b_0 , \ldots , b_{n-1} ) : C [a / x]}{
      \infer*{\Gamma \vdash b_{n-1} : C[ (n-1)_n / x ]}{
        \deduce{\Gamma \vdash b_0 : C[ 0_n / x ]}{
          \deduce{\Gamma \vdash a : \rN_n}{\Gamma , x : \rN_n \vdash C \; \type}
        }
      }
    }
    \quad
    \infer[\rN_n \eq]{\Gamma \vdash \rE_{n} ( m_n , b_0 , \ldots , b_{n-1} ) = b_{m} : C [ m_n / x]}{
      \infer*{\Gamma \vdash b_{n-1} : C[ (n-1)_n / x ]}{
        \deduce{\Gamma \vdash b_0 : C[ 0_n / x ]}{\Gamma , x : \rN_n \vdash C \; \type}
      }
    }
    \]

  \item[Rules for $\rW$-Types] a $\rW$-type $\rW_{(x : A)} B$ corresponds to the type of well founded trees such that each of their branchings is indexed by $B [a / x]$ for some $a : A$. Its elimination rule is the induction principle on a well-founded tree.
    \[
    \infer[\rW \fm]{\Gamma \vdash \rW_{(x : A)} B \; \type}{
      \deduce{\Gamma , x : A \vdash B \; \type}{\Gamma \vdash A \; \type}
    }
    \qquad
    \infer[\rW \intro]{\Gamma \vdash \wcons ( a , b ) : \rW_{(x : A)} B}{
      \deduce{\Gamma \vdash b : B[a/x] \to \rW_{(x : A)} B}{
        \deduce{\Gamma \vdash a : A}{\Gamma , x : A \vdash B \; \type}
      }
    }
    \]
    \[
    \infer[\rW \el]{\Gamma \vdash \rE_{\rW} ( a , b ) : C [a / y']}{
      \deduce{\Gamma \vdash b : \Pi_{(y : A)}\Pi_{(z : B[y/x] \to \rW_{(x : A)} B)}(\Pi_{(v : B[y/x])} C[ z \: v  /y'] \to C [\wcons (y , z) /y'])}{
        \deduce{\Gamma \vdash a : \rW_{(x : A)} B}{
          \deduce{\Gamma , y' : \rW_{(x : A)} B \vdash C \; \type}{\Gamma , x : A \vdash B \; \type}
        }
      }
    }
    \]
    \[
    \infer[\rW \eq]{\Gamma \vdash \rE_{\rW} ( \wcons ( a , b ) , c ) = c \: a \: b \: (\lambda v . \rE_{\rW} ( \app ( b , v) , c )) : C [\wcons ( a , b ) / y']}{
      \deduce{\Gamma  \vdash c : \Pi_{(y : A)}\Pi_{(z : B[y/x] \to \rW_{(x : A)} B)}(\Pi_{(v : B[y/x])} C[z \: v /y'] \to C [\wcons (y , z) /y'])}{
        \deduce{\Gamma \vdash b : B[a / x] \to \rW_{(x : A)} B}{
          \deduce{\Gamma \vdash a : A}{
            \deduce{\Gamma , y' : \rW_{(x : A)} B \vdash C \; \type}{\Gamma , x : A \vdash B \; \type}
          }
        }
      }
    }
    \]

  \item[Rules for $\Id$-Types] an $\Id$-type $\Id ( A , a_1 , a_2 )$ corresponds to an identity statement which says that $a_1$ and $a_2$ are equal elements of $A$.
    \[
    \infer[\Id \fm]{\Gamma \vdash \Id ( A , a_1 , a_2 ) \; \type}{
      \deduce{\Gamma \vdash a_2 : A}{
        \deduce{\Gamma \vdash a_1 : A}{\Gamma \vdash A \; \type}
      }
    }
    \qquad
    \infer[\Id \intro]{\Gamma \vdash \refl ( a ) : \Id ( A , a , a )}{\Gamma \vdash a : A}
    \]
    \[
    \infer[\Id \el]{\Gamma \vdash \rE_{\Id} (d , c): C [ a / x , b / y , d / z]}{
      \deduce{\Gamma \vdash  d : \Id ( A , a , b )}{
        \deduce{\Gamma \vdash b : A}{
          \deduce{\Gamma \vdash a : A}{
            \deduce{\Gamma \vdash c : \Pi_{(v : A)} C[ v / x , v / y , \refl ( v ) / z]}{\Gamma , x : A , y : A , z : \Id ( A , x , y ) \vdash C \; \type}
          }
        }
      }
    }
    \]
    \[
    \infer[\Id \eq]{\Gamma \vdash \rE_{\Id} ( \refl ( a ) , c ) = c \: a : C [a / x , a / y , \refl ( a ) / z]}{
      \deduce{\Gamma \vdash a : A}{
        \deduce{\Gamma \vdash c : \Pi_{(v : A)} C[ v / x , v / y , \refl ( v ) / z]}{\Gamma , x : A , y : A , z : \Id ( A , x , y ) \vdash C \; \type}
      }
    }
    \]

  \item[Rules for Mahlo Universe $( \rV , \rTV )$] we have the rules $\rV \fm$, $\rV \intro$ and $\rTV \fm$ with the inference rules which define the decoding function $\rTV$ recursively. Congruence rules are defined for the $\rTV$-formation rule $\rTV \fm$ and all $\rV$-introduction rules $\rV \intro$ except $\rV \intro_{\rN}$ and $\rV \intro_{\rN_n}$, though we omit to write down them.
    \[
    \infer[\rV \fm]{\Gamma \vdash \rV \; \type}{\vdash \Gamma \; \ctxt}
    \qquad
    \infer[\rTV \fm]{\Gamma \vdash \rTV (a) \; \type}{\Gamma \vdash a : \rV}
    \]
    \[
    \infer[\rV \intro_{\rN}]{\Gamma \vdash \wht{\rN}_{\rV} : \rV}{\vdash \Gamma \; \ctxt}
    \qquad
    \infer{\Gamma \vdash \rTV ( \wht{\rN}_{\rV} ) = \rN}{\vdash \Gamma \; \ctxt}
    \]
    \[
    \infer[\rV \intro_{\rN_n}]{\Gamma \vdash \wht{\rN_n}_{\rV} : \rV}{\vdash \Gamma \; \ctxt}
    \qquad
    \infer{\Gamma \vdash \rTV ( \wht{\rN_n}_{\rV} ) = \rN_n}{\vdash \Gamma \; \ctxt}
    \]
    \[
    \infer[\rV \intro_{\Pi}]{\Gamma \vdash \wht{\Pi}_{\rV} ( a , b ) : \rV}{
      \deduce{\Gamma \vdash b : \rTV ( a ) \to \rV}{\Gamma \vdash a : \rV}
    }
    \qquad
    \infer{\Gamma \vdash \rTV ( \wht{\Pi}_{\rV} ( a , b ) ) = \Pi_{(x : \rTV ( a ))} \rTV ( b \: x ) }{
      \deduce{\Gamma \vdash b : \rTV ( a ) \to \rV}{\Gamma \vdash a : \rV}
    }
    \]
    \[
    \infer[\rV \intro_{\Sigma}]{\Gamma \vdash \wht{\Sigma}_{\rV} ( a , b ) : \rV}{
      \deduce{\Gamma \vdash  b : \rTV ( a ) \to \rV}{\Gamma \vdash a : \rV}
    }
    \qquad
    \infer{\Gamma \vdash \rTV ( \wht{\Sigma}_{\rV} ( a , b ) ) = \Sigma_{(x : \rTV ( a ))} \rTV ( b \: x) }{
      \deduce{\Gamma \vdash b : \rTV ( a ) \to \rV}{\Gamma \vdash a : \rV}
    }
    \]
    \[
    \infer[\rV \intro_{+}]{\Gamma \vdash a \wht{+}_{\rV} b : \rV}{
      \deduce{\Gamma \vdash b : \rV}{\Gamma \vdash a : \rV}
    }
    \qquad
    \infer{\Gamma \vdash \rTV ( a \wht{+}_{\rV} b ) = \rTV ( a ) + \rTV ( b ) }{
      \deduce{\Gamma \vdash b : \rV}{\Gamma \vdash a : \rV}
    }
    \]
    \[
    \infer[\rV \intro_{\rW}]{\Gamma \vdash \wht{\rW}_{\rV} ( a , b ) : \rV}{
      \deduce{\Gamma \vdash b : \rTV ( a ) \to \rV}{\Gamma \vdash a : \rV}
    }
    \qquad
    \infer{\Gamma \vdash \rTV ( \wht{\rW}_{\rV} ( a , b ) ) = \rW_{(x : \rTV ( a ))} \rTV ( b \: x) }{
      \deduce{\Gamma \vdash b : \rTV ( a ) \to \rV}{\Gamma \vdash a : \rV}
    }
    \]
    \[
    \infer[\rV \intro_{\Id}]{\Gamma \vdash \wht{\Id}_{\rV} ( a , b , c ) : \rV}{
      \deduce{\Gamma \vdash c : \rTV ( a )}{
        \deduce{\Gamma \vdash b : \rTV ( a )}{\Gamma \vdash a : \rV}
      }
    }
    \qquad
    \infer{\Gamma \vdash \rTV ( \wht{\Id}_{\rV} ( a , b , c ) ) = \Id ( \rTV ( a ) , b , c  ) }{
      \deduce{\Gamma \vdash c : \rTV ( a )}{
        \deduce{\Gamma \vdash b : \rTV ( a )}{\Gamma \vdash a : \rV}
      }
    }
    \]
    The rules concerning the codes of non-Mahlo universes are as follows:
    \[
    \infer[\rV \intro_{\rU}]{\Gamma \vdash \wht{\rU}_{f} : \rV}{
      \Gamma \vdash f : \Sigma_{(x : \rV)} (\rTV ( x ) \to \rV) \to \Sigma_{(x : \rV)} (\rTV ( x ) \to \rV)
    }
    \]
    \[
    \infer[\rTV \mathsf{uni}]{\Gamma  \vdash \rTV ( \wht{\rU}_{f} ) = \rU_{f}}{
      \Gamma \vdash f : \Sigma_{(x : \rV)} (\rTV ( x ) \to \rV) \to \Sigma_{(x : \rV)} (\rTV ( x ) \to \rV)
    }
    \]
    \[
    \infer[\rV \intro_{\wht{\rT}}]{\Gamma  \vdash \wht{\rT}_{f} (a) : \rV}{
      \deduce{\Gamma  \vdash a : \rU_{f}}{
        \Gamma \vdash f : \Sigma_{(x : \rV)} (\rTV ( x ) \to \rV) \to \Sigma_{(x : \rV)} (\rTV ( x ) \to \rV)
      }
    }
    \]

  \item[Rules for Non-Mahlo Universes $( \rU_{f} , \wht{\rT}_{f} )$] below we define the $\rU_{f}$-formation rule $\rU_{f} \fm$, the $\rU_{f}$-introduction rules $\rU_{f} \intro$ and the inference rules which define the decoding function $\wht{\rT}_{f}$ recursively. Congruence rules are defined as in the case of the Mahlo Universe rules.

    We denote the assumption $\Gamma \vdash f : \Sigma_{(x : \rV)} (\rTV ( x ) \to \rV) \to \Sigma_{(x : \rV)} (\rTV ( x ) \to \rV)$ by $( \ast )$, and put $\rT_{f} (a) := \rTV (\wht{\rT}_{f} (a))$.
    \[
    \infer[\rU_{f} \fm]{\Gamma \vdash \rU_{f} \; \type}{( \ast )}
    \qquad
    \infer[\rU_{f} \intro_{\rN}]{\Gamma \vdash \wht{\rN}_{f} : \rU_{f}}{( \ast )}
    \qquad
    \infer{\Gamma \vdash \wht{\rT}_{f} ( \wht{\rN}_{f} ) = \wht{\rN}_{\rV} : \rV}{( \ast )}
    \]
    \[
    \infer[\rU_{f} \intro_{\Pi}]{\Gamma \vdash \wht{\Pi}_{f} ( a , b ) : \rU_{f}}{
      \deduce{\Gamma \vdash b : \rT_{f} ( a ) \to \rU_{f}}{
        \deduce{\Gamma \vdash a : \rU_{f}}{(\ast )}
      }
    }
    \qquad
    \infer{\Gamma \vdash \wht{\rT}_{f} ( \wht{\Pi}_{f} ( a , b ) ) = \wht{\Pi}_{\rV} ( \wht{\rT}_{f} ( a ) , \lambda x . \wht{\rT}_{f} ( b \: x) ) : \rV}{
      \deduce{\Gamma \vdash  b : \rT_{f} ( a ) \to \rU_{f}}{
        \deduce{\Gamma \vdash a : \rU_{f}}{(\ast )}
      }
    }
    \]
    \[
    \infer[\rU_{f} \intro_{\Id}]{\Gamma \vdash \wht{\Id}_{f} ( a , b , c ) : \rU_{f}}{
      \deduce{\Gamma \vdash c : \rT_{f} ( a )}{
        \deduce{\Gamma \vdash b : \rT_{f} ( a )}{
          \deduce{\Gamma \vdash a : \rU_{f}}{(\ast )}
        }
      }
    }
    \qquad
    \infer{\Gamma \vdash \wht{\rT}_{f} ( \wht{\Id}_{f} ( a , b , c ) ) = \wht{\Id}_{\rV} ( \wht{\rT}_{f} ( a ) , b , c  ) : \rV}{
      \deduce{\Gamma \vdash c : \rT_{f} ( a )}{
        \deduce{\Gamma \vdash b : \rT_{f} ( a )}{
          \deduce{\Gamma \vdash a : \rU_{f}}{(\ast )}
        }
      }
    }
    \]
    Similarly, we define the introduction and decoding function rules for $\wht{\rN_n}_{f} , \wht{\Sigma}_{f} , \wht{+}_{f} , \wht{\rW}_{f}$, but we omit to write down them.

    The following rules are the rules for the restriction of $f$ to $\rU_{f}$:
    \[
    \infer[\rU_{f} \intro_{\mathrm{res}^0}]{\Gamma \vdash \res{0}{f} ( a , b ) : \rU_{f}}{
      \deduce{\Gamma \vdash b : \rT_{f} ( a ) \to \rU_{f}}{
        \deduce{\Gamma \vdash a : \rU_{f}}{( \ast )}
      }
    }
    \quad
    \scalebox{0.90}{
    \infer{\Gamma \vdash \wht{\rT}_{f} ( \res{0}{f} ( a , b ) ) = \rp_1 (f \: (\wht{\rT}_{f} ( a ) , \lambda x . \wht{\rT}_{f} (b \: x))) : \rV}{
      \deduce{\Gamma \vdash b : \rT_{f} ( a ) \to \rU_{f}}{
        \deduce{\Gamma \vdash a : \rU_{f}}{( \ast )}
      }
    }
    }
    \]
    \[
    \infer[\rU_{f} \intro_{\mathrm{res}^1}]{\Gamma \vdash \res{1}{f} ( a , b , c ) : \rU_{f}}{
      \deduce{\Gamma \vdash c : \rT_{f} ( \res{0}{f} ( a , b ))}{
        \deduce{\Gamma \vdash b : \rT_{f} ( a ) \to \rU_{f}}{
          \deduce{\Gamma \vdash a : \rU_{f}}{( \ast )}
        }
      }
    }
    \quad
    \scalebox{0.90}{
    \infer{\Gamma \vdash \wht{\rT}_{f} ( \res{1}{f} ( a , b , c ) ) = \rp_2 (f \: (\wht{\rT}_{f} ( a ) , \lambda x . \wht{\rT}_{f} ( b \: x))) \: c : \rV}{
      \deduce{\Gamma \vdash c : \rT_{f} ( \res{0}{f} ( a , b ) )}{
        \deduce{\Gamma \vdash b : \rT_{f} ( a ) \to \rU_{f}}{
          \deduce{\Gamma \vdash a : \rU_{f}}{( \ast )}
        }
      }
    }
    }
    \]

  \end{description}
\end{defi}

Note that, in the definition above, we defined subuniverses of the Mahlo Universe $\rV$ as \textit{recursive subuniverses} in the sense of \cite{setzer2008,setzer2008a}: the decoding function $\wht{\rT}_{f}$ for the subuniverse $\rU_{f}$ was defined recursively.


\section{Accessibility Predicate in Extensional Type Theory}\label{sec:accproof}
This appendix shows that one can define the accessibility predicate $\Acc$ by using $\rW$-types if the core type theory of $\MLM$ is replaced with the extensional one in \cite{martinloef1984}. We denote this extensional variant of $\MLM$ by $\MLMext$. In fact, the argument below can be repeated in extensional $\MLTT$ of \cite{martinloef1984} itself. So, in $\MLMext$, one can simulate our interpretation of \textbf{Pi-Numbers} without adding $\Acc$ to the core type theory. None of the results given by this appendix is used in the other sections of the present paper. For comparison, recall that the core type theory of \cite{RathjenGrifforPalmgren1998} is the extensional type theory of \cite{martinloef1984} too, and that the interpretation of \textbf{Pi-Numbers} in \cite{RathjenGrifforPalmgren1998} does not use $\Acc$ at all.

First, we note that the transitive closure function $\mathsf{tc}$ enables to formulate a stronger transfinite induction on $\alpha : \bbV$ in \textit{intensional} $\MLTT$. Since we have this transfinite induction principle also in extensional $\MLTT$, we can use it to define the predicate $\Acc$ by using $\rW$-types in $\MLMext$.
\begin{lem}[Transfinite Induction on Transitive Closure, \href{https://github.com/takahashi-yt/czf-in-mahlo/blob/67c9dbf50bee67b73c1ea009929f4838458b8edc/src/CZFAxioms.agda\#L973}{\texttt{Link}}]\label{lem:tctrans}
  Let $F$ be a type with $\Gamma , \beta : \bbV \vdash F \; \type$. We then have a term
  \[
  \TI^{\mathsf{tc}} : \Pi_{(\alpha : \bbV)}(\forall \gamma \in \tc{\alpha} F[\gamma / \beta ] \to F[\alpha / \beta ]) \to \Pi_{(\alpha : \bbV)} F[\alpha / \beta ] .
  \]
\end{lem}

We replace the rules $\Id \intro$, $\Id \el$ and $\Id \eq$ with $\Id^e \intro$, $\Id^e \el$ and $\Id^e \eq$, respectively.
\[
\infer[\Id^e \intro]{\Gamma \vdash \erefl : \Id (A , a , b)}{
  \Gamma \vdash a = b : A
}
\qquad
\infer[\Id^e \el]{\Gamma \vdash a = b : A}{
  \deduce{\Gamma \vdash d : \Id (A , a , b)}{
    \deduce{\Gamma \vdash b : A}{
      \Gamma \vdash a : A
    }
  }
}
\qquad
\infer[\Id^e \eq]{\Gamma \vdash d = \erefl : \Id (A , a , b)}{
  \deduce{\Gamma \vdash d : \Id (A , a , b)}{
    \deduce{\Gamma \vdash b : A}{
      \Gamma \vdash a : A
    }
  }
}
\]
We call the resulting $\Id$-types the \textit{extensional} $\Id$-\textit{types}.

As is well known, the $\eta$-rule for $\Pi$-types with the rules for extensional $\Id$-types implies the principle of function extensionality.
\begin{description}
\item[Function Extensionality] for any $A , B$ with $\Gamma , x : A \vdash B \; \type$,
\[
\Gamma \vdash \Pi_{(f : \Pi_{(x : A)} B)}\Pi_{(g : \Pi_{(x : A)} B)} (\Pi_{(x : A)} f \: x =_{B} g \: x \to f =_{\Pi_{(x : A)} B} g)
\]
holds.
\end{description}
We then obtain the propositional induction principle for $\TI^{\mathsf{tc}}$, which we state as the lemma below. Note that \textbf{Function Extensionality} is the only extensional principle on which the lemma depends; so the lemma holds also in intensional $\MLTT$ with it.
\begin{lem}[Propositional Induction Principle for $\TI^{\mathsf{tc}}$, \href{https://github.com/takahashi-yt/czf-in-mahlo/blob/67c9dbf50bee67b73c1ea009929f4838458b8edc/src/Extensional.agda\#L104}{\texttt{Link}}]\label{lem:tcind}
  Let $F$ be a type with $\Gamma , \beta : \bbV \vdash F \; \type$. For any $g : \Pi_{(\alpha : \bbV)}(\forall \gamma \in \tc{\alpha} F[\gamma / \beta ] \to F[\alpha / \beta ])$ and any $\alpha : \bbV$, we have
  \[
  \TI^{\mathsf{tc}} \: g \: \alpha =_{F[\alpha / \beta ]} g \: \alpha \: (\lambda x . \TI^{\mathsf{tc}} \: g \: (\pred \: \tc{\alpha} \: x)) .
  \]
\end{lem}

By the propositional induction principle above, one can derive the following computation rule for $\TI^{\mathsf{tc}}$ in extensional $\MLTT$:
\[
\infer{\Gamma \vdash \TI^{\mathsf{tc}} \: g \: \alpha = g \: \alpha \: (\lambda x . \TI^{\mathsf{tc}} \: g \: (\pred \: \tc{\alpha} \: x)) : F[\alpha / \beta ]}{
  \deduce{\Gamma \vdash \alpha : \bbV}{
    \deduce{\Gamma \vdash g : \Pi_{(\alpha : \bbV)}(\forall \gamma \in \tc{\alpha} F[\gamma / \beta ] \to F[\alpha / \beta ])}{
      \Gamma , \beta : \bbV \vdash F \; \type
    }
  }
}
\]
This computation rule enables to define the accessibility predicate $\Acc$. Let $F$ be $\rV$, and $g_{\Acc}$ be
\[
\lambda \alpha . \lambda f . \wht{\Pi}_{\rV} (\idx \: \tc{\alpha} , \lambda x . f \: x) : \Pi_{(\alpha : \bbV)} (\forall \beta \in \tc{\alpha} \rV \to \rV) ,
\]
then we define $\Acc \: \gamma := \rTV (\TI^{\mathsf{tc}} \: g_{\Acc} \: \gamma)$.

Next, though we do not use this lemma below, we show in extensional $\MLTT$ that for every $\alpha : \bbV$, any two terms of type $\Acc \: \alpha$ are propositionally equal, hence they are also judgementally equal.
\begin{lem}[Extensional $\MLTT$]\label{lem:accprop}
  Let $\alpha$ be of type $\bbV$. We then have $d =_{\Acc \: \alpha} e$ for any two terms $d , e$ of type $\Acc \: \alpha$.
\end{lem}
\begin{proof}
  By transfinite induction on transitive closure (Lemma \ref{lem:tctrans}). The induction hypothesis says that
  \[
  \Pi_{(x : \ov{\tc{\alpha}})} \Pi_{(d' : \Acc \: (\pred \: \tc{\alpha} \: x))} \Pi_{(e' : \Acc \: (\pred \: \tc{\alpha} \: x))} d' =_{\Acc \: (\pred \: \tc{\alpha} \: x)} e'
  \]
  holds. Consider any arbitrary $d , e : \Acc \: \alpha$, and we show that $d =_{\Acc \: \alpha} e$ holds, where we have $d , e : \forall \beta \in \tc{\alpha} \Acc \: \beta$ by the $\TI^{\mathsf{tc}}$-computation rule:
  \begin{align*}
    \Acc \: \alpha &= \rTV (\TI^{\mathsf{tc}} \: g_{\Acc} \: \alpha ) = \rTV (g_{\Acc} \: \alpha \: (\lambda x . \TI^{\mathsf{tc}} \: g_{\Acc} \: (\pred \: \tc{\alpha} \: x))) \\
    &= \Pi_{(x : \ov{\tc{\alpha}})} \rTV (\TI^{\mathsf{tc}} \: g_{\Acc} \: (\pred \: \tc{\alpha} \: x)) = \Pi_{(x : \ov{\tc{\alpha}})} \Acc \: (\pred \: \tc{\alpha} \: x) = \forall \beta \in \tc{\alpha} \Acc \: \beta .
  \end{align*}
  By \textbf{Function Extensionality}, it suffices to verify that $\Pi_{(x : \ov{\tc{\alpha}})} d \: x =_{\Acc \: (\pred \: \tc{\alpha} \: x)} e \: x$ holds, but this claim follows from IH.
\end{proof}

In the remainder of this section, we derive the introduction, elimination and computation rules for $\Acc$ (cf. Definition \ref{def:acc}) in extensional $\MLTT$.
\begin{prop}[Extensional $\MLTT$]
  There are two terms $\prog$ and $\rE_{\Acc}$ satisfying the rules below.
  \[
  \infer[\intro]{\Gamma \vdash \prog \: f : \Acc \: \alpha}{
    \deduce{\Gamma \vdash f : \Pi_{(x : \ov{\tc{\alpha}})} \Acc \: (\pred \: \tc{\alpha} \: x)}{\Gamma \vdash \alpha : \bbV}
  }
  \]
  \[
  \infer[\el]{\Gamma \vdash \rE_{\Acc} \: \alpha \: t \: g : C [\alpha / \beta , t / x]}{
    \deduce{\Gamma \vdash g : \Pi_{(\gamma : \bbV)}\Pi_{(f : \Pi_{(y : \ov{\tc{\gamma}})} \Acc \: (\pred \: \tc{\gamma} \: y))}((\Pi_{(y : \ov{\tc{\gamma}})} C[\pred \: \tc{\gamma} \: y / \beta , f \: y / x]) \to C[\gamma / \beta , \prog \: f / x])}{
      \deduce{\Gamma , \beta : \bbV , x : \Acc \: \beta \vdash C \; \type}{
        \deduce{\Gamma \vdash t : \Acc \: \alpha}{\Gamma \vdash \alpha : \bbV}
      }
    }
  }
  \]
  \[
  \infer[\eq]{\Gamma \vdash \rE_{\Acc} \: \alpha \: (\prog \: f) \: g = g \: \alpha \: f \: (\lambda y . \rE_{\Acc} \: (\pred \: \tc{\alpha} \: y)  \: (f \: y) \: g ) : C [\alpha / \beta , \prog \: f / x]}{
    \deduce{\Gamma \vdash g : \Pi_{(\gamma : \bbV)}\Pi_{(f : \Pi_{(y : \ov{\tc{\gamma}})} \Acc \: (\pred \: \tc{\gamma} \: y))}((\Pi_{(y : \ov{\tc{\gamma}})} C[\pred \: \tc{\gamma} \: y / \beta , f \: y / x]) \to C[\gamma / \beta , \prog \: f / x])}{
      \deduce{\Gamma , \beta : \bbV , x : \Acc \: \beta \vdash C \; \type}{
        \deduce{\Gamma \vdash f : \Pi_{(y : \ov{\tc{\alpha}})} \Acc \: (\pred \: \tc{\alpha} \: y)}{\Gamma \vdash \alpha : \bbV}
      }
    }
  }
  \]
\end{prop}
\begin{proof}
  As shown in the above proof of Lemma \ref{lem:accprop}, we have $\Acc \: \alpha = \Pi_{(x : \ov{\tc{\alpha}})} \Acc \: (\pred \: \tc{\alpha} \: x)$. So we can define $\prog$ as the identity function $\lambda f . f$ on $\Pi_{(x : \ov{\tc{\alpha}})} \Acc \: (\pred \: \tc{\alpha} \: x)$, and then obtain the introduction rule for $\Acc$.

  Next, we derive the $\Acc$-elimination rule. We first provide a term of type
  \[
  \Pi_{(\alpha : \bbV)} \Pi_{(c : \Acc \: \alpha)} C [\alpha / \beta , c / x].
  \]
  Defining $F := \lambda \alpha . \Pi_{(c : \Acc \: \alpha)} C [\alpha / \beta , c / x]$, we use Lemma \ref{lem:tctrans}, namely, transfinite induction on transitive closure to construct a term of this type. Take variables $\alpha , c$ and $g'$ of types
  \begin{align*}
    &\alpha : \bbV , \\
    &c : \Acc \: \alpha \text{ and } \\
    &g' : \Pi_{(\gamma : \bbV)}\Pi_{(f : \Pi_{(y : \ov{\tc{\gamma}})} \Acc \: (\pred \: \tc{\gamma} \: y))}(\Pi_{(y : \ov{\tc{\gamma}})} C[\pred \: \tc{\gamma} \: y / \beta , f \: y / x] \to C[\gamma / \beta , \prog \: f / x]) ,
  \end{align*}
  respectively. The induction hypothesis is the following assumption:
  \[
  f' : \forall \gamma \in \tc{\alpha} \Pi_{(d : \Acc \: \gamma)} C [\gamma / \beta , d / x] .
  \]
  Since we have $c : \Acc \: \alpha = \forall \gamma \in \tc{\alpha} \Acc \: \gamma$, it follows from IH that
  \[
  \lambda y . f' \: y \: (c \: y) : \Pi_{(y : \ov{\tc{\alpha}})} C[\pred \: \tc{\alpha} \: y / \beta , c \: y / x]
  \]
  holds, and we denote the term $\lambda y . f' \: y \: (c \: y)$ by $s$. We thus obtain $g' \: \alpha \: c \: s : C [\alpha / \beta , \prog \: c / x]$, and we also have
  \[
  u := \lambda \alpha . \lambda f' . \lambda c . g' \: \alpha \: c \: s : \Pi_{(\alpha : \bbV)} \Pi_{(f' : \forall \gamma \in \tc{\alpha} \Pi_{(d : \Acc \: \gamma)} C [\gamma / \beta , d / x])} \Pi_{(c : \Acc \: \alpha )} C [\alpha / \beta , c / x]
  \]
  because $\prog \: c$ is judgementally equal to $c$. It thus follows that
  \[
  \TI^{\mathsf{tc}} \: u : \Pi_{(\alpha : \bbV)} \Pi_{(c : \Acc \: \alpha)} C [\alpha / \beta , c / x]
  \]
  holds with the judgemental equality
  \[
  \TI^{\mathsf{tc}} \: u \: \alpha = u \: \alpha \: (\lambda x . \TI^{\mathsf{tc}} \: u \: (\pred \: \tc{\alpha} \: x)) : \Pi_{(c : \Acc \: \alpha)} C [\alpha / \beta , c / x] .
  \]

  Define $\rE_{\Acc} := \lambda \alpha ' . \lambda c ' . \lambda g' . \TI^{\mathsf{tc}} \: u \: \alpha ' \: c '$. It is obvious that $\rE_{\Acc} \: \alpha \: t \: g$ is of type $C [\alpha / \beta , t / x]$, hence we have $\Acc$-elimination rule. Moreover, we also obtain the $\Acc$-computation rule because the following equations hold judgementally:
  \begin{align*}
    \rE_{\Acc} \: \alpha \: (\prog \: f) \: g &= \TI^{\mathsf{tc}} \: u[g / g'] \: \alpha \: (\prog \: f) = u[g / g'] \: \alpha \: (\lambda x . \TI^{\mathsf{tc}} \: u[g / g'] \: (\pred \: \tc{\alpha} \: x)) \: (\prog \: f) \\
    &= g \: \alpha \: (\prog \: f) \: (\lambda y . (\lambda x . \TI^{\mathsf{tc}} \: u[g / g'] \: (\pred \: \tc{\alpha} \: x)) \: y \: (\prog \: f \: y)) \\
    &= g \: \alpha \: (\prog \: f) \: (\lambda y . \TI^{\mathsf{tc}} \: u[g / g'] \: (\pred \: \tc{\alpha} \: y) \: (\prog \: f \: y)) \\
    &= g \: \alpha \: f \: (\lambda y . \TI^{\mathsf{tc}} \: u[g / g'] \: (\pred \: \tc{\alpha} \: y) \: (f \: y)) \\
    &= g \: \alpha \: f \: (\lambda y . (\lambda \alpha ' . \lambda c' . \lambda g' . \TI^{\mathsf{tc}} \: u \: \alpha ' \: c') \: (\pred \: \tc{\alpha} \: y) \: (f \: y) \: g) \\
    &= g \: \alpha \: f \: (\lambda y . \rE_{\Acc} \: (\pred \: \tc{\alpha} \: y) \: (f \: y) \: g) ,
  \end{align*}
  where the terms $\prog$ was defined as an identity function.
\end{proof}

\end{document}